\documentclass[smallextended,envcountsame]{svjour3}

\usepackage{amsmath,amssymb}
\usepackage{bbm}

\usepackage[numbers,sort&compress]{natbib}
\usepackage{enumitem}

\usepackage{nicefrac}

\usepackage[colorlinks,linkcolor=blue,citecolor=blue]{hyperref}

\makeatletter
\newcommand{\customlabel}[2]{%
   \protected@write \@auxout {}{\string \newlabel {#1}{{#2}{\thepage}{#2}{#1}{}} }%
   \hypertarget{#1}{#2}
}
\makeatother

\usepackage[capitalize]{cleveref}
\usepackage{thmtools}
\crefname{algocf}{Algorithm}{Algorithms}
\crefname{observation}{Observation}{Observations}

\spnewtheorem{observation}[theorem]{Observation}{\bfseries}{\itshape}
\spnewtheorem*{myclaim}{Claim}{\bfseries}{\itshape}

\usepackage{thm-restate}

\usepackage[algoruled,vlined]{algorithm2e}
\usepackage{setspace}

\def\mathrlap{\mathpalette\mathrlapinternal} 
\def\mathrlapinternal#1#2{\rlap{$\mathsurround=0pt#1{#2}$}}

\newcommand{\elsum}[1]{\sum_{\mathrlap{#1}}\;}

\usepackage{tikz}
\tikzstyle{normalNode}=[draw=black, fill=black, circle, minimum size=2mm, inner sep=0mm]
\tikzstyle{normalEdge}=[black, thick, >=stealth]

\newcommand{\st}{\,:\,}
\newcommand{\prob}[1]{\mathbf{Pr}\left[#1\right]}
\newcommand{\sets}{\mathcal{P}} 

\begin{document}

\title{Decomposition of Probability Marginals for Security Games in Max-Flow/Min-Cut Systems\footnote{A preliminary version of this work has appeared in the proceedings of IPCO 2023 under the title \emph{Decomposition of Probability Marginals for Security Games in Abstract Networks}~\citep{matuschke2023decomposition}.}}

\titlerunning{Decomposition of Probability Marginals in Max-Flow/Min-Cut Systems}

\author{Jannik Matuschke}

\institute{%
	Jannik Matuschke \at Research Center for Operations Management, KU Leuven;\\ \email{jannik.matuschke@kuleuven.be}
}%

\maketitle

\begin{abstract}
Given a set system $(E, \sets)$ with $\rho \in [0, 1]^E$ and $\pi \in [0,1]^{\sets}$, our goal is to find a probability distribution for a random set $S \subseteq E$ such that $\prob{e \in S} = \rho_e$ for all $e \in E$ and $\prob{P \cap S \neq \emptyset} \geq \pi_P$ for all $P \in \sets$.
We extend the results of \citet*{dahan2021probability} who studied this problem motivated by a security game in a directed acyclic graph~(DAG).

We focus on the setting where $\pi$ is of the affine form $\pi_P = 1 - \sum_{e \in P} \mu_e$ for~\mbox{$\mu \in [0, 1]^E$}.
A necessary condition for the existence of the desired distribution is that $\sum_{e \in P} \rho_e \geq \pi_P$ for all $P \in \sets$.
We show that this condition is sufficient if and only if $\sets$ has the \emph{weak max-flow/min-cut property}.
We further provide an efficient combinatorial algorithm for computing the corresponding distribution in the special case where~$(E, \sets)$ is an \emph{abstract network}. 
As a consequence, equilibria for the security game in~\cite{dahan2021probability} can be efficiently computed in a wide variety of settings (including arbitrary digraphs).

As a subroutine of our algorithm, we provide a combinatorial algorithm for computing shortest paths in abstract networks, partially answering an open question by \citet{mccormick1996polynomial}.
We further show that a conservation law proposed in~\cite{dahan2021probability} for the requirement vector $\pi$ in DAGs can be reduced to the setting of affine requirements described above.
\end{abstract}

\section{Introduction}
\label{sec:introduction}

Consider a set system $(E, \sets)$, where $E$ is a finite ground set and $\sets \subseteq 2^E$ is a collection of subsets of~$E$.
Given \emph{probability marginals} $\rho \in [0, 1]^E$ and \emph{requirements}
$\pi \in [0,1]^{\sets}$, we are interested in finding a probability distribution on the power set $2^E$ of $E$ that is consistent with these marginals and that ensures that each set in $P \in \sets$ is hit with probability at least $\pi_P$. In other words, we are looking for a solution~$x$ to the system
\begin{alignat}{3}
    \textstyle\sum_{S \subseteq E : e \in S} x_S & \;=\; \rho_e & \quad \forall\, e \in E,\label{eq:consistency}\\
    \textstyle\sum_{S \subseteq E : S \cap P \neq \emptyset} x_S & \;\geq\; \pi_P & \quad \forall\, P \in \mathcal{P},\label{eq:covering}\\
    \textstyle\sum_{S \subseteq E} x_S & \;=\; 1, &\label{eq:simplex}\\
    x_S & \;\geq\; 0 & \forall\, S \subseteq E. \label{eq:nonneg}
\end{alignat}
Throughout this paper, we will call a distribution $x$ fulfilling \eqref{eq:consistency} to \eqref{eq:nonneg} a \emph{feasible decomposition of $\rho$ for $(E, \sets, \pi)$}, and we will say that the marginals $\rho$ are \emph{feasible for $(E, \sets, \pi)$} if such a feasible decomposition exists.

A necessary condition for the existence of a feasible decomposition is that the marginals suffice to cover each set of the system individually, i.e., 
\begin{align}
    \textstyle \sum_{e \in P} \rho_e \geq \pi_P  \qquad \forall\, P \in \mathcal{P}. \tag{$\star$}\label{eq:cond}
\end{align}
We are particularly interested in identifying classes of systems and requirement functions for which \eqref{eq:cond} is not only a necessary but also a sufficient condition.
We say that $(E, \sets, \pi)$ is \emph{\eqref{eq:cond}-sufficient} if for all $\rho \in [0, 1]^E$ it holds that $\rho$ is feasible if and only if \eqref{eq:cond} is fulfilled.
For such systems, the set of distributions on $2^E$ fulfilling \eqref{eq:covering} can be described by the polytope of feasible marginals defined by the constraints~\eqref{eq:cond}, which is of exponentially lower dimension.

\subsection{Motivation}
\label{sec:motivation}

A natural application for feasible decompositions in the setting described above lies in network security games; see, e.g., \cite{bertsimas2016power,correa2017fare,dahan2021probability,holzmann2021shortest,szeszler2017security,tambe2011security} for various examples and applications of network security games.
In fact, such a game was also the motivation of \citet*{dahan2021probability}, who originally introduced the decomposition setting described above. We will discuss their game in detail in \cref{sec:game}. Here, we describe a simpler yet relevant problem as an illustrative~example.

Consider the following game played on a set system $(E, \sets)$, where
each element~$e \in E$ is equipped with a \emph{usage cost} $c_e \geq 0$ and an \emph{inspection cost}~\mbox{$d_e \geq 0$}. 
A \emph{defender} $D$ determines a random subset $S$ of elements from $E$ to inspect at cost $\sum_{e \in S} d_e$ (e.g., a set of links of a network at which passing traffic is monitored).
She anticipates that an \emph{attacker} $A$ is planning to carry out an illegal action, where $A$ chooses a set in $P \in \sets$ (e.g., a route in the network along which he smuggles contraband), for which he will receive utility~$U_1 - \sum_{e \in P} c_e$ for some constant $U_1 > 0$. However, if $P$ intersects with the random set $S$ of elements inspected by $D$, then $A$ is discovered while carrying out his illegal action, reducing his utility by a penalty $U_2 \geq U_1$.
The attacker also has the option to not carry out any attack, resulting in utility $0$. Thus, $A$ will refrain from using~$P \in \sets$ if the probability that $S \cap P \neq \emptyset$ exceeds $\pi_P := (U_1 - \sum_{e \in P} c_e) / U_2$.

A natural goal for $D$ is to discourage $A$ from attempting any attack at all,  while keeping the incurred inspection cost as small as possible.
Note that the randomized strategies that achieve this goal correspond exactly to vectors~$x$ that minimize $\sum_{S \subseteq E} \sum_{e \in S} d_e x_S$ subject to constraints~\eqref{eq:covering} to~\eqref{eq:nonneg}. 
Unfortunately, the corresponding LP has both an exponential number of variables and an exponential number of constraints in the size of the ground set $E$.

However, assume that we can establish the following three properties for our set system: 
\begin{itemize}
    \item[(1)] The system $(E, \sets, \pi)$ is \eqref{eq:cond}-sufficient.
    \item[(2)] We can efficiently compute feasible decompositions of given marginals.
    \item[(3)] We can efficiently solve $\min_{P \in \sets} \sum_{e \in P} \gamma_e$ for any given $\gamma \in \mathbb{R}_+^E$.
\end{itemize}
Then (1) allows us to formulate $D$'s problem in terms of marginals, i.e., $\min_{\rho \in [0,1]^E} \sum_{e \in E} d_e\rho_e$ subject to constraints \eqref{eq:cond}, (3) allows us to separate the linear constraints~\eqref{eq:cond} and obtain optimal marginals $\rho$, and (2) allows us to turn these marginals into a distribution corresponding to an optimal inspection strategy for~$D$.
In this article, we establish that the systems fulfilling condition~(1) are characterized by the so-called \emph{weak max-flow/min-cut property}, and in particular, we will show that all three conditions are fulfilled by a special class of systems with this property, called \emph{abstract networks}.

\subsection{The Weak Max-Flow/Min-Cut Property and Abstract Networks}
\label{sec:intro-mfmc}

A set system $(E, \sets)$ has the \emph{weak max-flow/min-cut (MFMC) property} if the polyhedron 
$$\textstyle Y_{\sets} := \left\{y \in \mathbb{R}_+^E \st \sum_{e \in P} y_e \geq 1 \ \forall\, P \in \sets\right\}$$ is integral.\footnote{%
The term \emph{weak max-flow/min-cut property} was introduced by \citet{seymour1977matroids} and distinguishes the property from the stronger \emph{max-flow/min-cut property}, which requires the system $\sum_{e \in P} y_e \geq 1$ for all $P \in \sets$ to be totally dual integral. 
A common assumption in the context of the weak MFMC property is that the  underlying set system is a \emph{clutter}, i.e., a set systems whose members are all inclusionwise minimal.
A clutter with the weak MFMC property is called \emph{ideal clutter}. It is easy to see that a set system $(E, \sets)$ has the weak MFMC property if and only if the induced clutter $\mathcal{C}(\sets) := \{P \in \sets \st Q \not\subset P \ \forall\, Q \in \sets\}$ is ideal.
However, all of our results hold for the general setting where $\sets$ is not necessarily a clutter.%
}
As the name suggests, the weak MFMC property indicates that, for any capacity vector $u \in \mathbb{R}^E$, the value of a maximum flow 
\begin{align}
    \textstyle \max_{f \in \mathbb{R}_+^{\sets}}\; \sum_{P \in \sets} f_P \quad \text{s.t. } \sum_{P : e \in P} f_P \leq u_e\ \forall\, e \in E
    \tag{MF} \label{eq:max-flow}
\end{align}
equals the capacity of a minimum cut
\begin{align}
    \textstyle \min_{C \subseteq E}\; \sum_{e \in C} u_e \quad \text{s.t. } P \cap C \neq \emptyset \ \forall\, P \in \sets.
    \tag{MC} \label{eq:min-cut}
\end{align}

Examples of systems with the weak MFMC property include the set of $s$-$t$-paths in a digraph, the set of arborescences in a digraph rooted at a fixed node, the set of $T$-joins in a graph with terminal set $T$, and the set of directed cuts in a digraph.
See the thesis of~\citet{abdi2018ideal} and the textbook by~\citet{cornuejols2001combinatorial} for an in-depth discussion, including all the aforementioned examples.

An \emph{abstract network} consists of a set system $(E, \sets)$ where each $P \in \sets$ (also referred to as \emph{abstract path}) is equipped with an internal linear order $\preceq_P$ of its elements, such that for all $P, Q \in \sets$ and all $e \in P \cap Q$ it holds that
\begin{align*}
    \exists R \in \sets \st R \,\subseteq\, \{p \in P \st p \preceq_P e \} \cup \{q \in Q \st e \preceq_Q q\}.
\end{align*}
Given two abstract paths $P, Q \in \sets$ and $e \in P \cap Q$, we use the notation $P \times_e Q$ to denote an arbitrary but fixed feasible choice for such an $R \in \sets$.

Intuitively, this definition is an abstraction of the property of digraphs that one can construct a new path by concatenating a prefix and a suffix of two intersecting paths.
Interesting special cases of abstract networks include the cases
where $\sets$ is the set of maximal chains in a partially ordered set $(E, \preceq)$,
and where $\sets$ is the set of simple $s$-$t$-paths in a digraph $D = (V, A)$, with $E = V \cup A$ and each path begin identified with the sequence of its nodes and arcs.
We remark that in both cases, the order $\preceq_{P \times_e Q}$ is consistent with $\preceq_P$ and $\preceq_Q$, which is not a general requirement for abstract networks; see, e.g.,~\cite{kappmeier2014abstract} for examples of abstract networks where this is not the case.  

Abstract networks were introduced by Hofmann~\cite{hoffman1974generalization} to illustrate the generality of Ford and Fulkerson's~\cite{ford1956maximal} original proof of the max-flow/min-cut theorem. 
He showed that abstract networks have the MFMC property by establishing total dual integrality for the dual of \eqref{eq:max-flow} when $(E, \sets)$ is an abstract network, even in a more general weighted version.

McCormick~\cite{mccormick1996polynomial} presented a combinatorial algorithm for solving problem~\eqref{eq:max-flow} when the abstract network is given by a \emph{membership oracle} that, given $F \subseteq E$, returns $P \in \sets$ with $P \subseteq F$ together with the corresponding order $\preceq_P$, or certifies that no such $P$ exists. Martens and McCormick~\cite{martens2008polynomial} later extended this result by giving a combinatorial algorithm for the aforementioned weighted version of the problem, using a stronger oracle.
Applications of abstract networks include, e.g., line planning for public transit systems~\cite{karbstein2013line} and route assignment in evacuation planning~\cite{pyakurel2022abstract,kappmeier2015generalizations}.

\subsection{Previous Results on Feasible Decompositions}

Dahan et al.~\cite{dahan2021probability} studied the decomposition setting described at the beginning of this article under the assumption that $\mathcal{P}$ is the set of maximal chains of a partially ordered set (poset), or, equivalently, the set of $s$-$t$-paths in a directed acyclic graph (DAG). 
They showed that \eqref{eq:cond} is sufficient for the existence of a feasible distribution when the requirements fulfil the conservation law
\begin{align}
    {\pi_P + \pi_Q = \pi_{P \times_e Q} + \pi_{Q \times_e P}} \quad \forall\, P, Q \in \mathcal{P}, e \in P \cap Q. \tag{C}\label{eq:conservation}
\end{align}
Although their result is algorithmic, the corresponding algorithm requires explicitly enumerating all maximal chains and hence does not run in polynomial time in the size of $E$.
However, \mbox{Dahan et al.~\cite{dahan2021probability}} provide a polynomial-time algorithm for the case of \emph{affine} requirements, in which there exists a vector $\mu \in [0, 1]^E$ such that the requirements are of the form
\begin{align}
    \textstyle \pi_P = 1 - \sum_{e \in P} \mu(e) \quad \forall\, P \in \sets, \tag{A}\label{eq:affine}
\end{align} 
and which constitutes a special case of \eqref{eq:conservation}.
As a consequence of this latter result, the authors were able to characterize Nash equilibria for their network security game (which is a flow-interdiction game played on $s$-$t$-paths in a digraph) by means of a compact arc-flow LP formulation and compute such equilibria in polynomial time, under the condition that the underlying digraph is acyclic.
Indeed, this positive result is particularly surprising, as similar---and seemingly simpler---flow-interdiction games had previously been shown to be NP-hard, even on DAGs~\cite{disser2020complexity}.

\subsection{Our Results}

We extend the results of Dahan et al.~\cite{dahan2021probability} for posets/DAGs in multiple directions:
\begin{enumerate}
    \item In \cref{sec:max-flow-min-cut}, we show that the weak MFMC property for $(E, \sets)$ is in fact equivalent to sufficiency of condition~\eqref{eq:cond} for all requirements $\pi$ fulfilling~\eqref{eq:affine}.
    Our proof is constructive and we show that a corresponding feasible decomposition can be computed in oracle-polynomial time when given access to a separation oracle for the polyhedron~$Y_{\sets}$. 
    \item In \cref{sec:feasible-decompositions}, we further provide a combinatorial, strongly polynomial algorithm for computing feasible decompositions when $(E, \sets)$ is an abstract network given by a membership oracle and $\pi$ fulfils~\eqref{eq:affine}. The algorithm is based on an explicit description of the decomposition via a natural generalization of shortest-path distances in abstract networks. 
    In particular, when $\sets$ is the set of $s$-$t$-paths in a digraph, a feasible decomposition can be computed efficiently by a single run of a standard shortest-path algorithm.
    In the case of generic abstract networks, our algorithm makes use of the following result as a subroutine.
    \item In \cref{sec:abstract-shortest-paths}, we provide a combinatorial, strongly polynomial algorithm for computing shortest paths in abstract networks, when $\sets$ is given by a membership oracle.
    Beyond its relevance for the present work,
    this algorithm is an interesting example of a routing algorithm that does not require any explicit information about local structures in the network, such as node-arc incidences.
    It also provides a necessary subroutine for a potential generalization of the Edmonds-Karp algorithm~\cite{edmonds1972theoretical} (which augments flow along shortest paths) to abstract networks, of which \citet{mccormick1996polynomial} conjectures that it might yield a strongly polynomial time algorithm for computing maximum flows in abstract networks.
    \item In \cref{sec:game} we discuss a consequence of our results: Nash equilibria for the network security game studied by \citet{dahan2021probability} can be described by a lower-dimensional LP and computed efficiently even when the game is played on a system with the weak MFMC property given via a separation oracle, or on an abstract network given via a membership oracle.
    In particular, this includes the case of a digraph with cycles. 
    \item In \cref{sec:conservation-law}, we show that the conservation law~\eqref{eq:conservation} proposed in \cite{dahan2021probability} for maximal chains in posets can be reduced to the affine requirements case~\eqref{eq:affine}. We provide a polynomial-time algorithm for computing the corresponding weights $\mu$ when the requirements $\pi$ are given by a value oracle. 
    As a consequence, the corresponding feasible decompositions can be computed efficiently in this case as well.%
    \footnote{Concurrently to the present work, \citet*{cela2023linear} derived an equivalent result in a different context, describing linearizable instances of higher-order shortest path problems; see \cref{sec:conservation-law} for details.}
    \item In \cref{sec:other-systems}, we turn our attention to systems that are not \eqref{eq:cond}-sufficient. We show that deciding whether a given marginal vector is feasible is NP-complete, even when $\sets$ is given by an explicit list of small sets and the requirements are all equal to~$1$.
    We also observe that approximately feasible decompositions always exist under condition \eqref{eq:cond}, i.e., decompositions in which each $P \in \sets$ is hit with a probability of at least $(1 - 1/\text{e}) \cdot \pi_P$.
\end{enumerate}

 \subsection{Notation}
Before we discuss our results in detail, we introduce some useful notation for abstract networks.
Let $(E, \sets)$ be an abstract network.
For $P \in \sets$ and $e \in P$, we use the following notation to denote prefixes of $P$ ending at $e$ and suffixes of $P$ starting at $e$, respectively:
\begin{alignat*}{3}
    [P, e] & \;:=\; \{p \in P \st p \preceq_P e \} & \qquad & [e, P] & \;:=\; \{p \in P \st e \preceq_P p \}\\
    (P, e) & \;:=\; \{p \in P \st p \prec_P e \} & \qquad & (e, P) & \;:=\; \{p \in P \st e \prec_P p \}
\end{alignat*}
For any path $P \in \sets$, we further denote the minimal and maximal element of $P$ with respect to $\preceq_P$ by $s_P$ and $t_P$, respectively.

\section{Feasible Decompositions and the Weak MFMC Property}
\label{sec:max-flow-min-cut}

In this section, we observe that $\eqref{eq:cond}$-sufficiency for arbitrary affine requirements can be characterized by the weak MFMC property.

\begin{theorem}\label{thm:mfmc}
Let $E$ be a finite set and $\sets \subseteq 2^E$. Then the following two statements are equivalent:
\begin{enumerate}
    \item $(E, \sets)$ has the weak MFMC property.\label{s:mfmc}
    \item $(E, \sets, \pi)$ is \eqref{eq:cond}-sufficient for all $\pi$ fulfilling \eqref{eq:affine}.\label{s:cond-sufficient}
\end{enumerate}
\end{theorem}

In the proof of \cref{thm:mfmc}, we describe how to construct a feasible decomposition for marginals satisfying~\eqref{eq:cond} from a convex combination of extreme points of the polyhedron $Y_{\sets}$. 
This convex combination can be obtained algorithmically if we have access to a \emph{separation oracle} for $Y_{\sets}$, that given $y \in \mathbb{R}^E$, either asserts that $y \in Y_{\sets}$ or finds a hyperplane separating $y$ from $Y_{\sets}$. 
This is formalized in the following theorem.

\begin{theorem}
\label{thm:mfmc-computation}
    There is an algorithm that, given a set system~$(E, \sets)$ with the weak MFMC property, accessed via a separation oracle for $Y_{\sets}$, and vectors \mbox{$\rho_e, \mu_e \in [0, 1]^E$} such that $\rho$ fulfils \eqref{eq:cond} for $\pi_P := 1 - \sum_{e \in P} \mu_e$, returns a feasible decomposition of $\rho$ with respect to $(E, \sets, \pi)$ in time polynomial in $T_Y$, $|E|$, and the encoding size of $\mu$ and $\rho$, where $T_Y$ denotes the time for a call to the separation oracle.
\end{theorem}

\paragraph{Separation for $Y_{\sets}$ equals Minimization for $(E, \sets)$.}
Note that the separation oracle for $Y_{\sets} = \left\{y \in \mathbb{R}_+^E \st \sum_{e \in P} y_e \geq 1 \ \forall\, P \in \sets\right\}$ is equivalent to an algorithm that, given $c \in \mathbb{R}_+^{E}$, finds $P \in \sets$ minimizing $\sum_{e \in P} c_e$, i.e., the separation problem for $Y_{\sets}$ can be solved efficiently if and only if one can efficiently minimize nonnegative linear objectives over $(E, \sets)$.

\subsection{Proof of \cref{thm:mfmc}}

Before we proove the theorem, we make the following helpful observation.

\begin{lemma}\label[lemma]{lem:feasible-min}
    Let $(E, \sets)$ be a set system and $\pi \in [0, 1]^{\sets}$.
    Let $\rho, \rho' \in [0, 1]^E$ with $\rho'_e \leq \rho_e$ for all $e \in E$. 
    Then $\rho$ is feasible if $\rho'$ is feasible.
    Moreover, given a feasible decomposition $x'$ for $\rho'$, a feasible decomposition for $\rho$ can be computed in time polynomial in $|E|$ and $|\{S \subseteq E \st x'_S > 0\}|$.
\end{lemma}
\begin{proof}
    Let $x'$ be a feasible decomposition of $\rho'$.
    Let $T' \subseteq E$ be a random set with $\prob{T' = S} = x'_S$ for all $S \subseteq E$.
    Furthermore let 
    \begin{align*}
        T'' := \left\{e \in E \st \rho_e > \rho'_e, \ \tau \leq \tfrac{\rho_e - \rho'_e}{1 - \rho'_e} \right\},
    \end{align*}
    where $\tau \sim U[0, 1]$ is drawn uniformly at random and independently from $T'$.
    Let $T := T' \cup T''$. Note that 
    \begin{align*}
        \prob{e \in T} = \prob{e \in T'} + \prob{e \notin T'}\prob{e \in T''} = \rho_e
    \end{align*}
    by construction of $T''$ and the fact that $\prob{e \in T'} = \rho'_e \leq \rho_e$.
    Note further that $\prob{P \cap T \neq \emptyset} \geq \prob{P \cap T' \neq \emptyset} \geq \pi_P$ for all $P \in \sets$, 
    because $x'$ is a feasible decomposition of $\rho'$.
    We conclude that $x_S := \prob{T = S}$ defines a feasible decomposition of $\rho$. 
    
    Note that there are at most $|E|$ different sets in the support of the distribution of $T''$, corresponding to the different threshold values for $\tau$ at which element $e \in E$ is excluded from $T''$.
    Therefore the support of the distribution of $T$ contains at most $|E| \cdot |\{S \subseteq E \st x'_S > 0\}|$ different sets and the above construction can be carried out in polynomial time. \qed
\end{proof}

\begin{proof}[\cref{thm:mfmc}]
    Let $\mathcal{S} := \{S \subseteq E \st P \cap S \neq \emptyset \ \forall, P \in \sets\}$ and denote by~$\mathbbm{1}_S \in \{0, 1\}^E$ the incidence vector of $S \subseteq E$.
    The following observations on the structure of $Y_{\sets}$ are useful for our proof.

    \begin{myclaim}
        Let $y$ be an extreme point of $Y_{\sets}$. Then $y_e \leq 1$ for all $e \in E$.
        Moreover, $Y_{\sets} \cap \{0, 1\}^E = \{\mathbbm{1}_S \in Y_{\sets} \st S \in \mathcal{S}\}$.
    \end{myclaim}
    \begin{proof}
        Let $y \in Y_{\sets}$ with $y_e > 1$ for some $e \in E$.
        Then $y + \varepsilon \mathbbm{1}_{\{e\}}, y - \varepsilon \mathbbm{1}_{\{e\}} \in Y_{\sets}$ for some $\varepsilon > 0$ and hence $y$ is not an extreme point.
        The second part of the claim follows directly from definition of $Y_{\sets}$ and $\mathcal{S}$.
        \quad $\blacklozenge$
    \end{proof} 

    We first show that Statement~\ref{s:mfmc} implies Statement~\ref{s:cond-sufficient}.
    Thus, let $\mu \in [0, 1]^E$ and define $\pi_P := 1 - \sum_{e \in P} \mu_e$ for $P \in \sets$.
    Let~$\rho \in [0, 1]^E$ fulfilling \eqref{eq:cond} for $(E, \sets, \pi)$.
    We will show that $\rho$ is feasible, using integrality of~$Y_{\sets}$.
    
    Define $y := \rho + \mu$. Note hat $y \in Y_{\sets}$ because $\sum_{e \in P} \rho_e + \mu_e \geq 1$ for all $P \in \sets$ by~\eqref{eq:cond}.
    By the claim above, integrality of $Y_{\sets}$ implies that every extreme point of $Y_{\sets}$ is of the form $\mathbbm{1}_S$ for some $S \in \mathcal{S}$.
    Furthermore, the extreme rays of $Y_{\sets}$ have only nonnegative components.
    Hence, we can write
    $y = \sum_{S \in \mathcal{S}} \lambda_S \mathbbm{1}_S + r$ for some $\lambda \in [0, 1]^{\mathcal{S}}$ with $\sum_{S \in \mathcal{S}} \lambda_S = 1$ and some $r \in \mathbb{R}_+^E$.
    
    We now define two random subsets of $E$ that will lead to a feasible decomposition of $\rho$.
    The first random set $T_1$ is defined by 
    $$\prob{T_1 = S} = \lambda_S \text{ for } S \in \mathcal{S}.$$
    Note that 
    $\prob{e \in T_1} = y_e - r_e \geq 0$ for all $e \in E$
    and $\prob{P \cap T_1} = 1$ for all $P \in \sets$ by construction of $T_1$.
    
    The second random set $T_2$ is independent from $T_1$ and defined by 
    $$T_2 := \left\{e \in E \st y_e > r_e, \ \tau < \tfrac{\rho_e}{y_e - r_e}\right\}$$ where $\tau \sim U[0, 1]$ is drawn uniformly at random from $[0, 1]$.
    Note that, by construction,
    $\prob{e \in T_2} = \min \left\{\rho_e / (y_e - r_e),\, 1\right\}$ for all $e \in E$ with $y_e - r_e > 0$.
    
    Now let $T := T_1 \cap T_2$. Note that also $T$ is a random set and that
    \begin{align*}
        \prob{e \in T} \;=\; \prob{e \in T_1} \prob{e \in T_2} \;=\; \min \{\rho_e,\, y_e - r_e\} \;\leq\; \rho_e
    \end{align*} 
    for all $e \in E$ by the above observations on the marginal probabilities for $T_1$ and $T_2$. Moreover, note that
    \begin{align*}
        \prob{P \cap T = \emptyset} & \;\leq\; \prob{\exists e \in P \cap T_1 : e \notin T_2}\\
        & \;\leq\; \sum_{e \in P} \prob{e \in T_1 \wedge e \notin T_2}\\
        & \;\leq\; \sum_{e \in P} (y_e - r_e) \left( 1 - \min \left\{\tfrac{\rho_e}{y_e - r_e}, 1\right\} \right)\\
        & \;\leq\; \sum_{e \in P} \max \{0,\; y_e - r_e - \rho_e\}  \;\leq\; \sum_{e \in P} \mu_e
    \end{align*}
     for all $P \in \sets$, where the first inequality follows from $\prob{P \cap T_1 \neq \emptyset} = 1$, the second inequality follows by union bound, the third inequality follows from independence of $T_1$ and $T_2$ and the observations on their marginal probabilities, the fourth inequality is a simple transformation using $y_e - r_e \geq 0$, and the fifth and final inequality follows from $\mu_e, r_e \geq 0$ and $y_e = \mu_e + \rho_e$. 
     
     We conclude that
        $\prob{P \cap T \neq \emptyset} \;\geq\; 1 - \sum_{e \in P} \mu_e \;=\; \pi_P$
    for all $P \in \sets$.
    Therefore, $T$ defines a feasible decomposition of the marginals $\rho'$ defined by $\rho'_e := \prob{e \in T} \leq \rho_e$, implying that $\rho$ is feasible by \cref{lem:feasible-min}. Thus $(E, \sets, \pi)$ is \eqref{eq:cond}-sufficient.
    
    We now show that Statement~\ref{s:cond-sufficient} implies Statement~\ref{s:mfmc}.
    Let $y$ be an extreme point of $Y_{\sets}$. 
    Note that $y_e \in [0, 1]^E$ by the claim at the beginning of the proof and that $y$ fulfils \eqref{eq:cond} for $\pi \equiv 1$. 
    Thus, $y$ has a feasible decomposition $x$ with respect to $(E, \sets, \pi)$ by Statement~\ref{s:cond-sufficient}. 
    Note that for any $S \subseteq E$ with $x_S > 0$ we must have $S \in \mathcal{S}$, because $1 = \sum_{S \subseteq E} x_S \geq \sum_{S : S \cap P \neq \emptyset} x_S \geq 1$ for all $P \in \sets$.
    Hence $y = \sum_{S \in \mathcal{S}} x_S \mathbbm{1}_S$, i.e., $y$ is a convex combination of the points $\mathbbm{1}_S$ for $S \in \mathcal{S}$. 
    Because $y$ is an extreme point of $Y$ and because $\mathbbm{1}_S \in Y_{\sets}$ for all $S \in \mathcal{S}$, we conclude that $y = \mathbbm{1}_S$ for some $S \in \mathcal{S}$ and hence $y$ is integral. \qed
\end{proof}

\subsection{Proof of \cref{thm:mfmc-computation}}

We now prove \cref{thm:mfmc-computation} by showing that the construction described in the preceding proof can be carried out efficiently using a separation oracle for $Y_{\sets}$.

\begin{proof}[\cref{thm:mfmc-computation}]
In order to proof the theorem, it suffices to show that an explicit description of the distributions of the sets $T_1$ and $T_2$ can be determined efficiently (in particular, we will see that the support of both distributions is of polynomial size).
From this we can compute the distribution of $T$ and then a feasible decomposition of $\rho$ by applying \cref{lem:feasible-min}.

Recall that $T_1$ is defined by $\prob{T_1 = S} = \lambda_{S}$ for $S \in \mathcal{S}$, where $\lambda \in [0, 1]^{\mathcal{S}}$ with $\sum_{S \in \mathcal{S}} \lambda_S = 1$ is part of a decomposition of the vector \mbox{$y = \mu + \rho \in Y_{\sets}$} of the form $y = \sum_{S \in \mathcal{S}} \lambda_S \mathbbm{1}_S + r$ for some $r \in \mathbb{R}^E_+$.
Given access to a separation oracle for $Y_{\sets}$, this decomposition of $y$ can be computed efficiently using the algorithmic version of Carath\'{e}odory's theorem~\citep[Corollary~6.5.13]{grotschel2012geometric}, with $\lambda_S > 0$ for at most $|E| + 1$ sets $S \in \mathcal{S}$.

Moreover, the distribution of the random set $T_2$ contains at most $|E|$ different sets in its support, corresponding to the threshold different values for $\tau$ at which element $e \in E$ is excluded from $T_2$.
We conclude that the distribution for $T$ has a support of size at most $|E| \cdot (|E| + 1)$ and that the support and corresponding probabilities can be computed in polynomial time.
\qed
\end{proof}

\section{Feasible Decompositions in Abstract Networks}
\label{sec:feasible-decompositions}

In this section we describe a combinatorial algorithm for computing feasible decompositions of marginals fulfilling \eqref{eq:cond} in abstract networks $(E, \sets)$ with affine requirements $\pi$. 
The algorithm is based on an explicit description of a feasible decomposition using shortest-path distances in abstract networks.
We present this description in \cref{sec:feasible-decompositions-description} and also show that for the special case that $\sets$ is the set of $s$-$t$-paths in a digraph, the decomposition can be computed using a single run of a standard shortest-path algorithm.
For the general case of arbitrary abstract networks, we then present a generic algorithm and its analysis in \cref{sec:computation}.

\subsection{Description of Feasible Decompositions using Shortest-Path Distances}
\label{sec:feasible-decompositions-description}

We will prove the following theorem, which describes a way to obtain feasible decompositions of marginals in abstract networks.

\begin{theorem}\label{thm:feasible-decomposition}
    Let $(E, \sets)$ be an abstract network and let $\rho, \mu \in [0, 1]^E$ fulfilling condition \eqref{eq:cond}, i.e., $\sum_{e \in P} \rho_e \geq \pi_P := 1 - \sum_{e \in P} \mu_e$ for all $P \in \sets$. Define
    \begin{align*}
        \textstyle \alpha'_e := \min \big\{\sum_{f \in (Q, e)} \mu_f + \rho_f \st Q \in \sets,\, e \in Q\big\} \ \text{and} \ \alpha_e := \min\, \{\alpha'_e,\, 1 - \rho_e\}
    \end{align*}
    for $e \in E$.
    For $\tau \sim U[0, 1]$ drawn uniformly at random from $[0, 1]$, let 
    \begin{align*}
        S_{\tau} := \{e \in E \st \alpha_e \leq \tau < \alpha_e + \rho_e\}.
    \end{align*}
    Then $x$ defined by $x_S := \prob{S_\tau = S}$ for $S \subseteq E$ is a feasible decomposition of~$\rho$ for $(E, \sets, \pi)$.
\end{theorem}

Intuitively, the values $\alpha'_e$ for $e \in E$ in the construction above correspond to the ``shortest-path distance'' to element $e$ in the abstract network $(E, \sets)$, with the truncation of $\alpha_e$ at $1 - \rho_e$ ensuring that $[\alpha_e, \alpha_e + \rho_e] \subseteq [0, 1]$.
Before we prove \cref{thm:feasible-decomposition}, let us first discuss some of its implications. 

\subsubsection{Interval Structure and Explicit Computation of $x$}

Given the vector~$\alpha$, the non-zero entries of $x$ can be easily determined in polynomial time.
Indeed, note that the set \mbox{$\Lambda := \{\alpha_e,\, \alpha_e + \rho_e \st e \in E\}$} induces a partition of $[0, 1]$ into at most~$2|E| + 1$ intervals (each with two consecutive values from $\Lambda \cup \{0, 1\}$ as its endpoints), such that $S_{\tau'} = S_{\tau''}$ whenever $\tau'$ and~$\tau''$ are in the same interval. Thus, there are at most~$2|E|+1$ non-zero entries in~$x$, whose values can be computed by sorting $\Lambda$, determining $S_{\tau}$ for some $\tau$ in each of the intervals induced by $\Lambda$, and then, for each occurring set~$S$, setting~$x_S$ to the total length of all intervals in which this set is attained.

\subsubsection{Special Case: Directed Graphs}

Consider the case where $\sets$ is the set of simple $s$-$t$-paths in a digraph $D = (V, A)$ and $E = V \cup A$. For $v \in V$, let $\sets_{sv}$ denote the set of simple $s$-$v$-paths in $D$. If we are given explicit access to $D$ (rather than accessing $\sets$ via a membership oracle), we can compute feasible decompositions as follows.
Without loss of generality, we can assume that for any $v \in V$ and $Q \in \sets_{sv}$ there is $Q' \in \sets_{st}$ with $Q \subseteq Q'$.\footnote{This can be ensured by introducing arcs $(v, t)$ with $\mu_{(v,t)} = 1$ and $\rho_{(v,t)} = 0$ for all $v \in V \setminus \{t\}$. Note that this does not change the set of feasible decompositions of $\rho$.}
Then $\alpha'_v = \min_{Q \in \sets_{sv}} \sum_{f \in Q \setminus \{v\}} \mu_f + \rho_f$ for $v \in V$ and $\alpha'_a = \min_{Q \in \sets_{sv}} \sum_{f \in Q} \mu_f + \rho_f$ for $a = (v, w) \in A$.
Hence, the vector~$\alpha'$ corresponds to shortest-path distances in $D$ with respect to $\rho+\mu$ (with costs on both arcs and nodes). Both $\alpha'$ and the corresponding feasible decomposition of $\rho$ can be computed by a single run of Dijkstra's~\cite{dijkstra1959note} algorithm.

\medskip

Computing feasible decompositions in the general case of arbitrary abstract networks is more involved. We show how this can be achieved in \cref{sec:computation}.

\subsubsection{Proof of \cref{thm:feasible-decomposition}}

The following lemma will be helpful for establishing the proof of the theorem.

\begin{lemma}\label[lemma]{lem:successor}
    Given $(E, \sets)$, $\rho$, $\mu$, and $\alpha$ as described in \cref{thm:feasible-decomposition},
   the following two conditions are fulfilled for every $P \in \sets$:
    \begin{enumerate}
        \item $\alpha_{t_P} + \mu_{t_P} + \rho_{t_P} \geq 1$ and \label{eq:t_p}
        \item for every $e \in P \setminus \{t_P\}$ there is $e' \in (e,P)$ with $\alpha_{e'} \leq \alpha_{e} + \mu_{e} + \rho_{e}$.\label{eq:succ}
    \end{enumerate} 
\end{lemma}
     
\begin{proof}
    We first show Statement~\ref{eq:t_p}.
    By contradiction assume $\alpha_{t_P} + \mu_{t_P} + \rho_{t_P} < 1$. Let $Q \in \sets$ with $t_P \in Q$ and $\sum_{f \in (Q, t_P)} \mu_f + \rho_f = \alpha_{t_P}$ and let $R := Q \times_{t_P} P$. Note that $R \subseteq [Q, t_P]$ and hence $\sum_{e \in R} \mu_e + \rho_e \leq \alpha_{t_P} + \mu_{t_P} + \rho_{t_P} < 1$, implying $\sum_{e \in R} \rho_{e} < 1 - \sum_{e \in R} \mu_{e}$, a contradiction to \eqref{eq:cond}.
    
    We now turn to Statement~\ref{eq:succ}. If $\alpha_e \geq 1 - \mu_e - \rho_e$, then the statement follows with $e' = t_P$ because $\alpha_{t_P} \leq 1 \leq \alpha_e + \mu_e + \rho_e$. Thus assume $\alpha_e < 1 - \mu_e - \rho_e$ and let $Q \in \sets$ with $\alpha_e =\sum_{f \in (Q,e)} \mu_{f} + \rho_{f}$.
    Let $R := Q \times_{e} P$.
    By \eqref{eq:cond} we observe that $\sum_{f \in R} \mu_{f} + \rho_{f} \geq 1 > \alpha_e + \mu_e + \rho_e$, which implies $R \setminus [Q, e] \neq \emptyset$ because $\mu, \rho \geq 0$.
    Thus, let $e' \in R \setminus [Q, e]$ be minimal with respect to $\prec_R$.
    Observe that $R \setminus [Q, e] \subseteq (e, P)$ and hence $e' \in (e, P)$. The statement then follows from
    \begin{align*}
        \alpha_{e'} & \textstyle \;\leq\; \sum_{f \in (R, e')} \mu_{f} + \rho_{f} \;\leq\; \sum_{f \in [Q, e]} \mu_{f} + \rho_{f} \;=\; \alpha_{e} + \mu_{e} + \rho_{e},
    \end{align*}
    where the second inequality is due to the fact that $(R, e') \subseteq [Q, e]$ by choice of $e'$ and the fact that $\mu, \rho \geq 0$. \qed
\end{proof}

We are now ready to proof the theorem.

\begin{proof}[\cref{thm:feasible-decomposition}]
We show that $x$ as described in \cref{thm:feasible-decomposition} is a feasible decomposition of $\rho$.
Note that $x$ fulfils \eqref{eq:simplex} and \eqref{eq:nonneg} by construction. Note further that $x$ fulfils \eqref{eq:consistency} because 
$$\elsum{S \subseteq E : e \in S} x_S = \prob{e \in S_\tau} = \prob{\alpha_e \leq \tau < \alpha_e + \rho_e} = \rho_e$$ for all $e \in E$, where the second identity follows from $0 \leq \alpha_e \leq 1 - \rho_e$.
With the help of \cref{lem:successor}, we can prove that $x$ fulfils \eqref{eq:covering} as follows. 

Let $P \in \sets$. 
     For $e \in P$ define 
     \begin{align*}
         \textstyle \phi(e) := \prob{S_\tau \cap [P, e] \neq \emptyset \;\wedge\; \tau \leq \alpha_e + \rho_e} \;+\; \sum_{f \in [P, e]} \mu_{f}.
     \end{align*}
     Let $F := \{e \in P : \phi(e) \geq \alpha_e + \mu_e + \rho_e \}$. 
     We will show that $t_P \in F$. 
     Note that this suffices to prove \eqref{eq:covering}, because the definition of $F$ together with Statement~\ref{eq:t_p} of \cref{lem:successor} imply $\phi(t_P) \geq \alpha_{t_P} + \mu_{t_P} + \rho_{t_P} \geq 1$, which in turn yields
     \begin{align*}
         \elsum{S \subseteq E : S \cap P \neq \emptyset} x_S \;=\; \prob{S_\tau \cap P \neq \emptyset} & \;\geq\; \phi(t_P) - \sum_{f \in P} \mu_{f} \;\geq\; 1 - \sum_{f \in P} \mu_{f} \;=\; \pi_P.
     \end{align*}
     
    We proceed to show $t_P \in F$.
    By contradiction assume this is not the case.
    Note that $F \neq \emptyset$ because $\alpha_{s_P} = 0$ and $\phi(s_P) = \prob{s_P \in S_\tau} + \mu_{s_P} = \rho_{s_P} + \mu_{s_P}$. 
    Thus let $e \in F$ be maximal with respect to $\prec_{P}$.
    Because $e \neq t_P$, we can invoke Statement~\ref{eq:succ} of \cref{lem:successor} and obtain $e' \in (e, P)$ with 
    \begin{align}
        \alpha_{e'} \leq \alpha_{e} + \mu_{e} + \rho_{e}.\label{eq:successor}
    \end{align}
    We will show that $e' \in F$, contradicting our choice of $e$.
    Note that the definition of $\phi$ and the fact that $e' \succ_P e$ imply
    \begin{align}
        \phi(e') & \;\geq\; \phi(e) \,+\, \prob{e' \in S_{\tau} \,\wedge\, \tau > \alpha_{e} + \rho_{e}} \,+\, \mu_{e'} \notag\\
        & \;\geq\; \alpha_{e} + \mu_{e} + \rho_{e} \,+\, \prob{e' \in S_{\tau} \,\wedge\, \tau > \alpha_{e} + \rho_{e}} + \mu_{e'}, \label{eq:abstract-induction}
    \end{align}
    where the second inequality follows from $e \in F$.
    Moreover, observe that $e' \in S_\tau$ if and only if $\alpha_{e'} \leq \tau < \alpha_{e'} + \rho_{e'}$ and hence
    \begin{align*}
         \prob{e' \in S_{\tau} \,\wedge\, \tau > \alpha_e + \rho_e} & \;=\; \alpha_{e'} + \rho_{e'} - \max \{\alpha_{e'},\, \alpha_e + \rho_{e}\}\\
         & \; \geq \; \alpha_{e'} + \rho_{e'} - (\alpha_e + \mu_e + \rho_e),
    \end{align*}
    where the inequality follows from \eqref{eq:successor}.
    Combining this bound with \eqref{eq:abstract-induction} yields $\phi(e') \geq \alpha_{e'} + \mu_{e'} + \rho_{e'}$ and hence $e' \in F$, contradicting our choice of $e$ and completing the proof of \cref{thm:feasible-decomposition}. \qed
\end{proof}

\subsection{A Combinatorial Algorithm for Computing Feasible Decompositions}
\label{sec:computation}

We now describe an algorithm for computing feasible decompositions in arbitrary abstract networks $(E, \sets)$ with affine requirements.
We assume that the ground set $E$ is given explicitly, while the set of abstract paths $\sets$ is given by a \emph{membership oracle} that, given $F \subseteq E$, either returns~$P \in \sets$ with $P \subseteq F$ and the corresponding order $\preceq_P$, or confirms that no $P \in \sets$ with $P \subseteq F$ exists.

By \cref{thm:feasible-decomposition}, it suffices to compute the values of $\alpha_e$ for all $e \in E$. Unfortunately, a complication arises in that even finding a path containing a certain element $e \in E$ is NP-hard.\footnote{Note that even for the special case where $\mathcal{P}$ corresponds to the set of simple $s$-$t$-paths in a digraph, finding $P \in \sets$ containing a certain arc $e$ is equivalent to the 2-disjoint path problem (for $P$ to be simple, its prefix up to the tail of $e$ and its suffix starting from the head of $e$ must be disjoint).  Simply side-stepping this issue by introducing additional elements as done in the second remark after \cref{thm:feasible-decomposition} is not possible here, because we are restricted to accessing $\sets$ only via the membership oracle.} 
However, as we show below, it is possible to identify a subset $U \subseteq E$ for which we can compute the values of $\alpha$, while the elements in $E \setminus U$ turn out to be redundant with respect to feasibility of the marginals.
From this, we obtain the following theorem.

\begin{theorem}\label{thm:computation}
    There is an algorithm that, given an abstract network $(E, \sets)$ via a membership oracle and $\rho, \mu \in [0, 1]^E$ such that $\sum_{e \in P} \rho_e \geq \pi_P := 1 - \sum_{e \in P} \mu_e$ for all $P \in \sets$, computes a feasible decomposition of $\rho$ for $(E, \sets, \pi)$ in time $\mathcal{O}(|E|^3 \cdot T_{\sets})$, where $T_\sets$ denotes the time for a call to the membership oracle.
\end{theorem}

\subsubsection{The Algorithm}
\cref{thm:computation} is established via \cref{alg:compute-decomposition}, which computes values $\bar{\alpha}_e$ for elements $e$ in a subset $U \subseteq E$ as follows.
Starting from $U = \emptyset$, the algorithm iteratively computes a path $P$ minimizing $\sum_{f \in P \cap U} \mu_f + \rho_f$ and adds the first element $e$ of $P \setminus U$ to $U$, determining $\bar{\alpha}_e$ based on the length of $(P, e)$.

\vspace{-0.3cm}

\begin{algorithm}[h]
  \caption{Computing a feasible decomposition}\label{alg:compute-decomposition}
  \setstretch{1.1}
  Initialize $U := \emptyset$.\\
  \While{$\min_{P \in \sets} \sum_{f \in P \cap U} \mu_f + \rho_f < 1$\vspace{0.1cm}}{
    Let $P \in \operatorname{argmin}_{P \in \sets} \sum_{f \in P \cap U} \mu_{f} + \rho_{f}$.\\
    Let $e := \min_{\preceq_{P}} P \setminus U$.\\
    Set $U := U \cup \{e\}$ and $\bar{\alpha}_e := \min \big\{\sum_{f \in (P,e)} \mu_{f} + \rho_{f},\; 1 - \rho_e\big\}$.
  }
  \Return $\bar{\alpha}, U$
\end{algorithm}

\vspace{-0.6cm}

\subsubsection{Analysis} 

First note that in every iteration of the while loop, the set $P \setminus U$ is nonempty because $\sum_{f \in P} \mu_f + \rho_f \geq 1$ by the assumption on the input in \cref{thm:computation}.
Hence the algorithm is well-defined and terminates after at most $|E|$ iterations. We further remark that finding $P \in \sets$ minimizing $\sum_{e \in P \cap U} \rho_f + \mu_f$ can be done in time $\mathcal{O}(|E|^2 T_{\sets})$ using the \cref{alg:shortest-path} described in \cref{sec:abstract-shortest-paths}. The following lemma (proven below) then suffices to complete the proof of
\cref{thm:computation}.

\begin{restatable}{lemma}{restateLemAlgDecomposition}\label[lemma]{lem:alg-decomposition}
  Let $\bar{\alpha}, U$ be the output of \cref{alg:compute-decomposition} and define $\bar{\rho}_e := \rho_e$ and $\bar{\mu}_e := \mu_e$ for $e \in U$ and $\bar{\rho}_e := 0$ and $\bar{\mu}_e := 0$ for $e \in E \setminus U$. Then
  \begin{enumerate}
      \item $\sum_{e \in P} \bar{\rho}_e \geq \bar{\pi}_P := 1 - \sum_{e \in P} \bar{\mu}_e$ for all $P \in \sets$ and\label{eq:U-sufficient}
      \item $\bar{\alpha}_e = \min \left\{\sum_{f \in (Q, e)} \bar{\mu}_{f} + \bar{\rho}_f \st Q \in \sets,\, e \in Q \right\} \cup \{1 - \bar{\rho}_e\}$ for all $e \in U$.\label{eq:alpha-prime}
  \end{enumerate}
\end{restatable}

Indeed, observe that \cref{lem:alg-decomposition} together with \cref{thm:feasible-decomposition} implies that $\bar{\alpha}$ induces a feasible decomposition $\bar{x}$ of $\bar{\rho}$ for $(E, \sets)$ and $\bar{\pi}$. Because $\bar{\rho}_e \leq \rho_e$ for all $e \in E$ and $\bar{\pi}_P \geq \pi_P$ for all $P \in \sets$, this decomposition can be extended to a feasible decomposition of $\rho$ for $(E, \sets)$ and $\pi$ by arbitrarily incorporating the elements from $E \setminus U$.
This completes the proof of \cref{thm:computation}.

\begin{proof}[\cref{lem:alg-decomposition}]
    Let $P^{(i)} \in \sets$ and $e^{(i)} \in E$ be the path and the element chosen in iteration $i$ of the while loop. Let $U^{(0)} := \emptyset$ and $U^{(i)} := \{e^{(1)}, \dots, e^{(i)}\}$. Note that $U^{(i-1)}$ is the state of $U$ at the beginning of iteration $i$ of the while loop.
    Note that \eqref{eq:cond} implies
    $\sum_{f \in P^{(i)}} \rho_{f} + \mu_{f} \geq 1 > \sum_{f \in P^{(i)} \cap U^{(i-1)}} \mu_{f} + \rho_{f}$, where the final inequality follows from the termination criterion of the while loop. Thus $P^{(i)} \setminus U^{(i-1)} \neq \emptyset$, implying that the second line of the while loop is well-defined and the algorithm adds a new element to $U$ in each iteration of the while loop. We conclude that the algorithm terminates after $k$ iterations for some $k \leq |E|$. 
    
    The termination criterion implies $\sum_{e \in P \cap U} \rho_e \geq 1 - \sum_{e \in P \cap U} \mu_e \geq \pi_P$ for all $P \in \sets$, proving statement \ref{eq:U-sufficient} of the lemma.
    We next establish statement \ref{eq:alpha-prime} of the lemma. To this end let $e \in U$ and $Q \in \sets$ with $e \in Q$. Let $i \in \{1, \dots, k\}$ be such that $e = e^{(i)}$ and define $R := Q \times_{e} P^{(i)}$.
    Note that 
    \begin{align*}
        \elsum{\qquad f \in P^{(i)} \cap U^{(i-1)}} \mu_{f} + \rho_{f} 
        \ \leq \ \elsum{\quad f \in R \cap U^{(i-1)}} \mu_{f} + \rho_{f} 
        \ \leq \ \elsum{\quad f \in (Q, e) \cap U^{(i-1)}} \mu_{f} + \rho_{f} \;+\; \elsum{\qquad f \in [e, P^{(i)}] \cap U^{(i-1)}} \mu_{f} + \rho_{f},
    \end{align*}
    where the first inequality follows from the choice of $P^{(i)}$ by the algorithm in iteration $i$, and the second inequality follows from construction of $R$.
    From this, we conclude 
     \begin{align*}
        \alpha_{e} 
        \ = \ \elsum{\quad f \in (P^{(i)}, e)} \mu_{f} + \rho_{f}
        \ = \ \elsum{\qquad f \in (P^{(i)}, e) \cap U^{(i-1)}} \mu_{f} + \rho_{f}
        \ \leq \ \elsum{\qquad f \in (Q, e) \cap U^{(i-1)}} \mu_{f} + \rho_{f}
        \ \leq \ \elsum{\quad f \in (Q, e)} \mu_{f} + \rho_{f},
    \end{align*}
    where the second identity follow from $ (P^{(i)}, e) \subseteq U^{(i-1)}$ by choice of $e = e^{(i)}$ as first element on $P^{(i)} \setminus U^{(i-1)}$. \qed
\end{proof}

\section{Computing Shortest Paths in Abstract Networks}
\label{sec:abstract-shortest-paths}

In this section, we consider the following natural generalization of the classic shortest $s$-$t$-path problem in digraphs: Given an abstract network $(E, \sets)$ and a cost vector $\gamma \in \mathbb{R}_+^E$, find a path $P \in \sets$ minimizing $\sum_{e \in P} \gamma_e$.
We provide a combinatorial, strongly polynomial algorithm for this problem, accessing~$\sets$ only via a membership oracle.

In fact, the question for such an algorithm was already raised by McCormick~\cite{mccormick1996polynomial}, who conjectured that it can be used to turn an adaptation of his combinatorial, but only weakly polynomial algorithm for the maximum abstract flow problem into a strongly polynomial one. Our result shows that such a shortest-path algorithm indeed exists, but leave it open how to use it to improve the running time of the maximum abstract flow algorithm.
Another interesting aspect of our result is that it shows that it is possible to efficiently compute shortest paths in digraphs, even when no explicit local information about the network is given, such as incidence lists of nodes is given.

\begin{theorem}\label{thm:shortest-path}
  There is an algorithm that, given an abstract network $(E, \sets)$ via a membership oracle and a vector $\gamma \in \mathbb{R}_+^E$, computes $P \in \sets$ minimizing $\sum_{e \in P} \gamma_e$
  in time $\mathcal{O}(|E|^2 \cdot T_{\sets})$, where $T_\sets$ denotes the time for a call to the membership oracle for $\sets$.
\end{theorem}

\subsection{The Algorithm}
For notational convenience, we assume that there is $s, t \in E$ with $s_P = s$ and $t_P = t$ for all $P \in \sets$. Note that this assumption is without loss of generality, as it can be ensured by adding dummy elements $s$ and $t$ to $E$ and including them at the start and end of each path, respectively.

The algorithm is formally described as \cref{alg:shortest-path}. It can be seen as a natural extension of Dijkstra's~\cite{dijkstra1959note} algorithm in that it maintains for each element \mbox{$e \in E$} a (possibly infinite) label $\psi_e$ indicating the length of the shortest segment $[Q_e, e]$ for some $Q_e \in \sets$ with $e \in Q_e$ found so far, and in that its outer loop iteratively chooses an element with currently smallest label for processing.
However, updating these labels is more involved, as an abstract network does not provide local concepts such as ``the set of arcs leaving a node''. In its inner loop, the algorithm therefore carefully tries to extend the segment $Q_e$ for the currently processed element $e$ to find new shortest segments $Q_{e'}$ for other elements $e'$.

\vspace{-0.3cm}

\begin{algorithm}[h]
  \caption{Computing a shortest path in an abstract network}\label{alg:shortest-path}
  \setstretch{1.1}
  Initialize $T := \emptyset$, $\psi_s := \gamma_s$, and $\psi_e := \infty$ for all $e \in E \setminus \{s\}$.\\
  Let $Q_s \in \sets$.\\
  \While{$\psi_t > \min_{f \in E \setminus T} \psi_f$\vspace{0.1cm}}{
    Let $e \in \operatorname{argmin}_{f \in E \setminus T} \psi_f$.\\
    Let $F := (E \setminus T) \cup [Q_e, e]$.\\
    \While{there is $P \in {\sets}$ with $P \subseteq F$\vspace{0.1cm}}{
      Let $e' := \min_{\preceq_{P}} P \setminus [Q_e, e]$.\\
      Set $F := F \setminus \{e'\}$.\\
      \If{$\sum_{f \in [P, e']} \gamma_f < \psi_{e'}$\vspace{0.1cm}}{
        Set $\psi_{e'} := \sum_{f \in [P, e']} \gamma_f$ and $Q_{e'} := P$.
      }
    }
    Set $T := T \cup \{e\}$.
  }
  \Return $Q_t$
\end{algorithm}

\vspace{-0.8cm}

\subsection{Analysis} 

The proof of the correctness of \cref{alg:shortest-path} crucially relies on the following lemma, which intuitively certifies that the algorithm does not overlook any shorter path segments when processing an element.

\begin{restatable}{lemma}{restateLemShortestPathInvariant}
\label[lemma]{lem:shortest-path-invariant}
    \cref{alg:shortest-path} maintains the following invariant: For all $P \in \sets$, there is $e \in P$ with $[e, P] \cap T = \emptyset$ and $\psi_e \leq \sum_{f \in [P, e]} \gamma_f$.
\end{restatable}

\begin{proof}
    The invariant is clearly fulfilled initially as $T = \emptyset$ and $\psi_s = \gamma_s = \sum_{e \in [P, s]} \gamma_e$ for all $P \in \sets$ (recall our assumption $s = s_P \in P$ for all $P \in \sets$).
    Now assume that the invariant holds at the beginning of an iteration of the outer while loop, i.e., for each $P \in \sets$ there is $e_P \in P$ with $[e_P, P] \cap T = \emptyset$ and $\psi(e_P) = \sum_{f \in Q_{e_P}} \gamma_f \leq \sum_{f \in [P, e_P]} \gamma_{f}$.
    Let $e \in \operatorname{argmin}_{f \in E \setminus T} \psi_f$ be the element selected at the beginning of this iteration and added to $T$ at its end.
    Consider any path $P \in \sets$ and distinguish two cases:
    \begin{description}
        \item Case 1: $e \notin [e_P, P]$. In this case, $e_P$ continues to fulfil the conditions of the invariant for $P$, as $[e_P, P] \cap (T \cup \{e\}) = \emptyset$ and $\psi_{e_P} = \sum_{f \in [Q_{e_P}, e_P]} \gamma_f$ either remains unchanged or is decreased during the course of the inner while loop.
        \item Case 2: $e \in [e_P, P]$. Let $R := Q_e \times_e P$. Observe that $R \subseteq E \setminus T \cup [Q_e, e]$ and hence there must be an iteration of the inner while loop in which some $e' \in R$ is chosen for removal of $F$ with $e'$ being $\prec_{P'}$-minimal in $P' \setminus [Q_e, e]$ for some $P' \in \sets$.
        In that iteration of the inner while loop, $Q_{e'}$ is set to $P'$ and $\psi_{e'}$ is set to 
        \begin{align*}
            \psi_{e'} = \elsum{f \in [P', e']} \gamma_f \leq \elsum{f \in [Q_e, e]} \gamma_f \;+\; \gamma_{e'} = \psi_e + \gamma_{e'} \leq \psi_{e_P} + \gamma_{e'},
        \end{align*}
        where the fist inequality follows from $[P', e'] \subseteq [Q_e, e] \cup \{e'\}$ by choice of $e'$ as first element of $P'$ not in $[Q_e, e]$, and the second inequality follows from the choice of $e \in \operatorname{argmin}_{f \in E \setminus T} \psi_f$ and the fact that $e_P \notin T$.
        Note that $e' \in R \setminus [Q_e, e]$ implies $e' \in (e, P)$ and
        $[e', P] \cap (T \cup \{e\}) = \emptyset$, as $(e, P) \subseteq [e_P, P]$.
        Note further that  $\psi_{e'} \leq \psi_{e_P} + \gamma_{e'} \leq \sum_{f \in [P, e']} \gamma_f$, from which we conclude that $e'$ fulfils the statement of the invariant for $P$. \qed
    \end{description} 
\end{proof}

We now show how the preceding lemma implies the correctness of \cref{alg:shortest-path}, thus proving \cref{thm:shortest-path}.

\begin{proof}[\cref{thm:shortest-path}]
    When \cref{alg:shortest-path} terminates, $\psi_t \leq \psi_f$ for all $f \in E \setminus T$ by the termination criterion of the outer while loop.
    Let $P \in \sets$.
    By \cref{lem:shortest-path-invariant} there is an element $e \in P \setminus T$ with $\psi_e \leq \sum_{f \in [P, e]} \gamma_f$.
    Note that this implies $\sum_{f \in Q_t} \gamma_f = \psi_t \leq \psi_e = \sum_{f \in [P, e]} \gamma_f \leq \sum_{f \in P} \gamma_f$,
    where the last inequality uses the fact that $\gamma_f \geq 0$ for all $f \in E$. 
    We conclude that the path $Q_t$ returned by the algorithm is indeed a shortest path.
    
    To see that the algorithm terminates in polynomial time, observe that 
    the outer while loop stops after at most $|E|-1$ iterations, as
    in each iteration an element from $E \setminus \{t\}$ is added to $T$ and the termination criterion is fulfilled if $T = E \setminus \{t\}$.
    Furthermore, each iteration of the inner while loop removes an element from $F$ and hence after at most $|F| \leq |E|$ iterations no path $P \subseteq F$ exists anymore, implying that the inner while loop terminates. \qed
\end{proof}

\section{Dahan et al.'s Network Security Game}
\label{sec:game}

Dahan et al.~\cite{dahan2021probability} studied the following network security game. The input is a set system $(E, \sets)$ with capacities $u \in \mathbb{R}_+^E$, transportation cost $c \in \mathbb{R}_+^E$ and interdiction costs $d \in \mathbb{R}_+^E$.
There are two players: the \emph{routing entity} $R$, whose strategy space is the set of \emph{flows} $F := \{f \in \mathbb{R}_+^\sets \st \sum_{P \in \sets : e \in P} f_P \leq u_e\  \forall e \in E\}$, and the \emph{interdictor} $I$, who selects a subset of elements $S \subseteq E$ to interdict, with the intuition that all flow on interdicted elements is disrupted.
Given strategies $f \in F$ and $S \subseteq E$, the payoffs for $R$ and $I$, respectively, are given by
\begin{align*}
    \Phi_R(f, S) & := \textstyle \sum_{P \in \sets : P \cap S = \emptyset} f_P - \sum_{P \in \sets} \sum_{e \in P} c_e f_P \text{ and}\\
    \Phi_I(f, S) & := \textstyle  \sum_{P \in \sets : P \cap S \neq \emptyset} f_P - \sum_{e \in S} d_e,
\end{align*}
respectively. 
That is, $R$'s payoff is the total amount of non-disrupted flow, reduced by the cost for sending flow $f$, while $I$'s payoff is the total amount of flow that is disrupted, reduced by the interdiction cost for the set $S$.

We are interested in finding (mixed) Nash equilibria (NE) for this game, i.e., random distributions $\sigma_R$ and $\sigma_I$ over the strategy spaces of $I$ and $R$, respectively, such that no player can improve their expected payoff by unilateral deviation. However, the efficient computation of such equilibria is hampered by the fact that the strategy spaces of both players are of exponential size/dimension in the size of the ground set $E$. 
To overcome this issue, Dahan et al.~\cite{dahan2021probability} considered the following pair of primal and dual linear programs:
\begin{alignat*}{6}
    [\text{LP}_R]~\max\ && \sum_{P \in \sets} \pi^c_P & f_P &&& \quad [\text{LP}_I]~\min\  && \sum_{e \in E} u_e \eta_e + & d_e\rho_e \\
    \text{s.t.}\ && \elsum{P \in \sets : e \in P} f_P & \leq u_e & \ \forall e \in E && 
    \qquad \text{s.t.}\ && \elsum{e \in P} \eta_e + \rho_e & \geq \pi^c_P & \ \forall P \in \sets\\
    && \elsum{P \in \sets : e \in P} f_P & \leq d_e & \ \forall e \in E && && \eta, \rho & \geq 0\\
    && f & \geq 0 & && && &
\end{alignat*}
where $\pi^c_P := 1 - \sum_{a \in P} c_a$. 
They showed the following result.
\begin{theorem}[Dahan et al.~\cite{dahan2021probability}]\label{thm:NE-dahan}
    Let $f^*$ and $(\eta^*, \rho^*)$ be optimal solutions to $[\text{LP}_R]$ and $[\text{LP}_I]$, respectively. Let $\sigma_I$ be a feasible decomposition of $\rho^*$ for $(E, \sets, \pi)$, where $\pi_P := \pi^c_P - \sum_{e \in P} \eta^*_e$, and let $\sigma_R$ be a distribution over $F$ with $\sum_{f \in F} \sigma_{R,f} f_P = f^*_P$. Then $(\sigma_R, \sigma_I)$ is a Nash equilibrium.
\end{theorem}
In particular, note that any feasible solution to $[\text{LP}_I]$ defines marginals $\rho$ that fulfil~\eqref{eq:cond} for $\pi_P := \pi^c_P - \sum_{e \in P} \eta_e = 1 - \sum_{e \in P} c_e + \eta_e$.
Hence, if $(E, \sets, \pi)$ is \eqref{eq:cond}-sufficient for all affine requirements $\pi$, then any pair of optimal solutions to the LPs induces a Nash equilibrium.
If we can moreover efficiently compute optimal solutions to the LPs and the corresponding feasible decompositions, we can efficiently find a Nash equilibrium. 

Dahan et al.~\cite{dahan2021probability} showed that these conditions are met when $\sets$ is the set of $s$-$t$-paths in a DAG. 
Thus, NE for the game can be found efficiently in that setting. 
This positive result is particularly interesting because NE are NP-hard to compute for the variant of the game in which the interdictor is limited by a budget, even when interdiction costs are uniform, transportation costs are zero, and the game is played on a DAG~\cite{disser2020complexity}.

\subsection{Implications of Our Results}

As established by \cref{thm:mfmc}, the systems that fulfil \eqref{eq:cond}-sufficiency for all affine $\pi$ are exactly the systems with the weak MFMC property.
Thus, by \cref{thm:NE-dahan}, the linear programs $[\text{LP}_R]$ and $[\text{LP}_I]$ describe Nash equilibria for security games played in such systems.

If we can moreover efficiently minimize nonnegative linear objectives over a such a system, then we can separate the constraints of $[\text{LP}_I]$, obtaining optimal solutions to both LPs, and, finally, as shown by \cref{thm:mfmc-computation}, we can efficiently compute the corresponding feasible decomposition, and thus compute a Nash equilibrium.
This yields the following theorem.
\begin{theorem}
    There is an algorithm that given an instance of the security game of Dahan et al.~\citep{dahan2021probability} on a set system $(E, \sets)$ with the weak MFMC property, and a separation oracle for $Y_\sets$, computes a Nash equilibrium in time polynomial in $|E|$, the encoding size of $c$, $d$, and $u$ and $T_{\sets}$, the time for a call to the separation oracle.
\end{theorem}
In particular, our results imply efficient computability of NE when $\sets$ is the set of $s$-$t$-paths in a (not necessarily acyclic) digraph, when $\sets$ is an abstract network given by a membership oracle, or when $\sets$ is any of systems with the weak MFMC property mentioned in \cref{sec:intro-mfmc} (arborescenses, $T$-joins, directed cuts).

\subsection{Characterizing All Nash Equilibria}

\citet{dahan2021probability} also identified situations in which the linear programs $[\text{LP}_R]$ and $[\text{LP}_I]$ describe the set of \emph{all} NE for their security game. 
Specifically, they showed that every NE can be constructed from optimal solutions to the LPs as described in \cref{thm:NE-dahan}, if in addition to \eqref{eq:cond}-sufficiency of $(E, \sets, \pi)$ for all affine $\pi$, 
there exists an optimal solution $\eta^*, \rho^*$ to $[\text{LP}_I]$ such that $\rho^*$ has a feasible decomposition $x$ with resepect to~$\pi_P := 1 - \sum_{e \in P} c_e + \eta^*_e$ with $x_{\emptyset} > 0$.

They observed that this latter condition is indeed fulfilled when $(E, \sets)$ is the set of $s$-$t$-paths in a DAG and $c_e > 0$ for all $e \in E$.
Here, we argue that the same is true for arbitrary set systems $(E, \sets)$ with the weak MFMC property.

Assume $c_e > 0$ for all $e \in E$ and let $(\eta^*, \rho^*)$ be an optimal solution to~$[\text{LP}_I]$. Note that we can assume without loss of generality that $\rho^*$ is minimal in the sense that for every $e \in E$ with $\rho^* > 0$, there exists a $P \in \sets$ with $\sum_{f \in P} \eta^*_f + \rho^*_f = \pi^c_P$.
Defining $\mu_e := \eta^*_e + c_e > 0$, the following lemma, which is a consequence of the proof of \cref{thm:mfmc}, shows that such a minimal $\rho^*$ has a decomposition $x$ with $x_{\emptyset} > 0$.

\begin{lemma}\label[lemma]{lem:emptyset}
    Let $(E, \sets)$ be a set system with the weak MFMC property and let $\mu_e \in [0, 1]^E$ with $\mu_e > 0$ for all $e \in E$.
    If $\rho \in [0, 1]$ fulfils \eqref{eq:cond} for $(E, \sets, \pi)$ with $\pi_P := 1 - \sum_{e \in P} \mu_e$ and if for every $e \in E$ with $\rho_e > 0$ there is $P \in \sets$ with $e \in P$ and $\sum_{f \in P} \rho_f = \pi_P$, then there is a feasible decomposition $x$ of $\rho$ with respect to $(E, \sets, \pi)$ such that $x_{\emptyset} > 0$.
\end{lemma}

\begin{proof}
As in the proof of \cref{thm:mfmc}, let $y := \mu + \rho$ and consider an arbitrary representation of $y$ as a convex combination of extreme points of~$Y_{\sets}$ and a ray $r \geq 0$ of $Y_{\sets}$, i.e., 
$y = \sum_{S \in \mathcal{S}} \lambda_S \mathbbm{1}_S + r$ for some $\lambda \in [0, 1]^{\mathcal{S}}$ with $\sum_{S \in \mathcal{S}} \lambda_S = 1$ and some $r \in \mathbb{R}_+^E$.

We first show that $\rho_e > 0$ implies $\rho_e < 1$ and $r_e = 0$.
Indeed, by the premise of the lemma, there is $P \in \sets$ with $e \in P$ such that $\sum_{f \in P} \rho_f = 1 - \sum_{f \in P} \mu_f$.
Thus $\rho_e < \mu_e + \rho_e \leq \sum_{f \in P} \rho_f + \mu_f = 1$.
Moreover, because $y - r \in Y_{\sets}$, it holds that $\sum_{f \in P} \mu_f + \rho_f - r_f \geq 1$, which implies $r_f = 0$ for all $f \in P$ and thus in particular $r_e = 0$.

Now let $T_1$, $T_2$, and $T = T_1 \cap T_2$ be the random sets constructed in the proof of \cref{thm:mfmc}.
We show that $\prob{T_2 = \emptyset} = 1 - \max_{e \in E : y_e > r_e} \frac{\rho_e}{y_e - r_e} > 0$.
Indeed, let $e \in E$ with $y_e > r_e$ and note that if $\rho_e > 0$, then $r_e = 0$ and hence $\frac{\rho_e}{y_e - r_e} = \frac{\rho_e}{\mu_e + \rho_e} < 1$ because $\mu_e > 0$.

As shown in the proof of \cref{thm:mfmc}, $x_S := \prob{T = S}$ defines a feasible decomposition of the marginal vector $\rho'$ defined by $\rho'_e := \prob{e \in T} \leq \rho_e$. Moreover, because for every $e \in E$ with $\rho_e > 0$ there is $P \in \sets$ with $e \in P$ and $\sum_{f \in P} \rho'_f \leq \sum_{f \in P} \rho_f = \pi_P$, we concluded that $\rho = \rho'$.
Finally, note that $x_{\emptyset} \geq \prob{T_2 = \emptyset} > 0$ by the observation above.
\qed
\end{proof}

Hence the LPs $[\text{LP}_R]$ and $[\text{LP}_I]$ indeed describe all Nash equilibria of the game when $c_e > 0$ for all $e \in E$.
We conjecture this is still true under a weaker assumption, namely that $\pi^c_P < 1$ for all $P \in \sets$. However, it seems that this requires a different construction of the decomposition than the one given in \cref{thm:mfmc}.

\section{The Conservation Law for Partially Ordered Sets}
\label{sec:conservation-law}

As discussed in \cref{sec:introduction}, Dahan et al.~\cite{dahan2021probability} established the sufficiency of \eqref{eq:cond} in partially ordered sets not only for the case of affine requirements~\eqref{eq:affine} but also for the case where requirements fulfil the conservation law~\eqref{eq:conservation}. However, they left it open whether it is possible to efficiently compute the corresponding decompositions in the latter case.
In this section, we resolve this question by showing that the conservation law~\eqref{eq:conservation} for maximal chains in a poset can be reduced to the case of affine requirements~\eqref{eq:affine} in the corresponding Hasse diagram\footnote{The Hasse diagram of a poset $(E, \preceq)$ is a directed acyclic graph whose nodes are the elements of $E$ and whose maximal paths correspond to the maximal chains of~$\preceq$. See \cref{sec:hasse} for details on this transformation and why it is necessary.}, for which a feasible decomposition then can be computed efficiently.

\begin{restatable}{theorem}{restateThmConservationAffine}\label{thm:conservation-affine}
Let $D = (V, A)$ be a DAG, let $s, t \in V$, and let $\sets \subseteq 2^{V \cup A}$ be the set of $s$-$t$-paths in $D$.
Let $\pi \in [0, 1]^\sets$ such that \eqref{eq:conservation} is fulfilled.
Then there exists $\mu \in [0,1]^{V \cup A}$ such that $\pi_P = 1 - \sum_{e \in P} \mu_e$. 
Furthermore, $\mu$ can be computed in strongly polynomial time in $|V|$ and $|A|$, and $T_{\pi}$, where $T_{\pi}$ denotes the time for a call to an oracle that, given $P \in \sets$, returns $\pi_P$.
\end{restatable}

Independently from the present work, a result equivalent to \cref{thm:conservation-affine} was recently proven by \citet[Theorem~1]{cela2023linear} in the context of characterizing linearizable instances of higher-order shortest path problems.
Their proof is based on an inductive (along a topological order in the DAG) description of the weights $\mu$, which yields an algorithm for computing $\mu$ in time~$\mathcal{O}(|A|^2 + |A| T_{\pi})$, as we describe in detail in \cref{sec:quadratic-mu} below.

\begin{proof}[\cref{thm:conservation-affine}] 
Consider the system 
    \begin{alignat}{3}
        \textstyle \sum_{e \in P} \mu_e & = 1 - \pi_P &\qquad \forall\, P \in \sets,\label{eq:affine-transform}\\
        \mu_e & \geq 0 &\qquad \forall\, e \in E := V \cup A,\label{eq:affine-nonneg}
    \end{alignat}
and note that any feasible solution $\mu$ to this system satisfies the conditions established in the theorem, as $\mu_e \leq 1$ for all $e \in E$ is implied by \eqref{eq:affine-transform} and \eqref{eq:affine-nonneg}.

For any vector $y \in \mathbb{R}^{\sets}$, define
\begin{align*}
    \textstyle \bar{\pi}(y) := \sum_{P \in \sets} (1 - \pi_P) y_P \quad \text{ and } \quad y_e  := \sum_{P \in \sets : e \in P} y_P \text{ for } e \in E.
\end{align*}
Using this notation, Farkas' lemma states that system \eqref{eq:affine-transform} and \eqref{eq:affine-nonneg} has a feasible solution if and only if 
\begin{align}
    \bar{\pi}(y) \geq 0
    \text{ for all } y \in \mathbb{R}^{\sets} \text{ with } y_e \geq 0 \text{ for all } e \in E. \label{eq:farkas}
\end{align}

To show that \eqref{eq:farkas} is fulfilled, we use the following lemma, which we prove in \cref{sec:existence-mu}.

\begin{lemma}\label[lemma]{lem:farkas-transform}
    Let $z, z' \in \mathbb{R}_+^{\sets}$ with $z_e \geq z'_e$ for all $e \in E$.
    Then $\bar{\pi}(z) \geq \bar{\pi}(z')$.
\end{lemma}

Let $y \in \mathbb{R}^{\sets}$ with $y_e \geq 0$ for all $e \in E$. We define $z, z' \in \mathbb{R}_+^{\sets}$ by setting $z_P := \max \{y_P, 0\}$ and $z'_P := \max \{-y_P, 0\}$ for $P \in \sets$. Note that $z_e \geq z'_e$  for all $e \in E$ because $z_e - z'_e = y_e \geq 0$. Hence \cref{lem:farkas-transform} applied to $z$ and $z'$ yields $\bar{\pi}(y) = \bar{\pi}(z) - \bar{\pi}(z') \geq 0$. This proves \eqref{eq:farkas} and thus the existence of $\mu$.

To show that $\mu$ can be computed in polynomial time, it suffices to determine a basis of the rows of the system~\eqref{eq:affine-transform}.
The following lemma, which we prove in \cref{sec:computing-mu}, shows how to obtain such a basis efficiently.

\begin{lemma}\label[lemma]{lem:computing-mu}
  There is an algorithm that, given a DAG \mbox{$D = (V, A)$} and two nodes $s, t, \in V$, computes in strongly polynomial time a set $\bar{\sets}$ of $s$-$t$-paths in~$D$ such that the incidence vectors of the paths in~$\bar{\sets}$ are a basis of the subspace of~$\mathbb{R}^{V \cup A}$ spanned by the incidence vectors of all $s$-$t$-paths in~$D$.
\end{lemma}
Given the basis, a solution $\mu$ fulfilling \eqref{eq:affine-transform} and \eqref{eq:affine-nonneg} can be computed in strongly polynomial time, e.g., using the algorithm of \citet{tardos1986strongly} for linear programs with binary constraint matrices. 
This completes the proof of the theorem.
\qed
\end{proof}

\subsection{Proof of \cref{lem:farkas-transform} (Farkas Conditions)}
\label{sec:existence-mu}

\begin{proof}[\cref{lem:farkas-transform}]
Let $z' \in \mathbb{R}_+^{\sets}$. 
Let $$X_{z'} := \{z \in \mathbb{R}_+^{\sets} : z_e \geq z'_e \text{ for all } e \in E\}.$$ 
We show that $\min_{z \in X_{z'}} \bar{\pi}(z) \geq \bar{\pi}(z')$, which proves the lemma.
We show this by induction on the cardinality of $$\operatorname{supp}(z') := \{P \in \sets : z'_P > 0\}.$$ 

To start the induction, note that if $\operatorname{supp}(z') = \emptyset$, then $\bar{\pi}(z') = 0 \leq \bar{\pi}(z)$ for all $z \in \mathbb{R}_+^{\sets}$, because $1 - \pi_P \geq 0$ for all $P \in \sets$.
To complete the induction, we make use of the following claim, which we prove further below.
\begin{myclaim}
    There is $z \in \operatorname{argmin}_{\bar{z} \in X_z'} \bar{\pi}(\bar{z})$ such that $z_P \geq z'_P > 0$ for some $P \in \sets$.
\end{myclaim}
Using the claim, we can construct $\hat{z}, \hat{z}' \in \mathbb{R}_+^{\sets}$ by $\hat{z}_P := z_P - z'_P \geq 0$, $\hat{z}'_P := 0$, and $\hat{z}_{\bar{P}} := z_{\bar{P}}$, $\hat{z}'_{\bar{P}} := z'_{\bar{P}}$ for $\bar{P} \in \sets \setminus \{P\}$.
Note that $\hat{z}_e \geq \hat{z}'_e$ for all $e \in E$ and 
 $\operatorname{supp}(\hat{z}') = \operatorname{supp}(z') \setminus \{P\}$ by construction. Hence
\begin{align*}
    \bar{\pi}(z) - (1 - \pi(P)) z'_P = \bar{\pi}(\hat{z}) \geq \bar{\pi}(\hat{z}') = \bar{\pi}(z') - (1 - \pi(P)) z'_P
\end{align*}
by induction hypothesis. We conclude that $\bar{\pi}(z) \geq \bar{\pi}(z')$, which completes the induction, as $z$ is a minimizer of $\bar{\pi}$ in $X_{z'}$.

To complete the proof of the lemma, it remains to prove the claim.
For this, fix an arbitrary $P \in \sets$ with $z'_P > 0$. 
For any $a \in P \cap A$ and any $z \in X_{z'}$, define 
    $\mathcal{Q}(a) := \{Q \in \sets : a \in Q \text{ and } [Q, a] = [P, a]\}$ and $q_z(a) := \sum_{Q \in \mathcal{Q}(a)} z_P$.
Let $z$ be a maximizer of $\sum_{a \in P \cap A} q_{z}(a)$ among all $z \in \operatorname{argmin}_{\bar{z} \in X_{z'}} \bar{\pi}(\bar{z})$. 

We will show that $z_P \geq z'_P$. 
By contradiction assume that this is not the case.
Let $a = (u, v) \in P \cap A$ be the last arc on $P$ such that $q_{z}(a) \geq z'(P)$,
and let $a' = (v, w) \in P \cap A$ be the first arc on $P$ such that $q_{z}(a) < z'(P)$ (note that $a$ and $a'$ exist because the first arc of $P$ fulfils $q_{z}(a) \geq z'_P$ but the last arc of $P$ does not, as $z_P < z'_P$).

Let $Q \in \mathcal{Q}(a)$ with $a' \notin Q$ and $z_Q > 0$
and let $R \in \sets \setminus \mathcal{Q}(a)$ with $a' \in R$ and $z_R > 0$ (note that $Q$ exists because $q_{z}(a') < z'_P \leq q_{z}(a)$ by choice of $a$ and $a'$, and $R$ exists because $z_a \geq z'_P > q_{z}(a)$).
By construction, $v \in Q \cap R$. Thus, let $\bar{Q} := Q \times_v R$, $\bar{R} := R \times_v Q$ and $\varepsilon := \min \{z_Q, z_R\} > 0$. We construct $\bar{z}$ by 
\begin{align*}
    \bar{z}(\bar{P}) := \begin{cases}
            z_{\bar{P}} + \varepsilon & \text{ if } \bar{P} \in \{\bar{Q}, \bar{R}\},\\
            z_{\bar{P}} - \varepsilon & \text{ if } \bar{P} \in \{Q, R\},\\
            z_{\bar{P}} & \text{ otherwise}.
        \end{cases}
\end{align*}
Note that $\bar{z} \in \mathbb{R}_+^{\sets}$ by choice of $\varepsilon$. Furthermore, $\bar{z}_e = z_e$ for all $e \in E$, because the graph is acyclic and hence $\bar{Q} = [Q, v] \cup [v, R]$ and $\bar{R} = [R, v] \cup [v, Q]$.
Finally, observe that $\bar{\pi}(\bar{z}) = \bar{\pi}(z)$ by \eqref{eq:conservation}. 
These three observations imply that $\bar{z}$ is contained in $X_{z'}$ and it is a minimizer of $\bar{\pi}$.
Moreover, $q_{\bar{z}}(\bar{a}) \geq q_z(\bar{a})$ for all $\bar{a} \in P \cap A$
because $Q \in \mathcal{Q}(\bar{a})$ implies $\bar{Q} \in \mathcal{Q}(\bar{a})$ and $R \in \mathcal{Q}(\bar{a})$ implies $\bar{R} \in \mathcal{Q}(\bar{a})$, respectively.
Furthermore, $\bar{Q} \in \mathcal{Q}(a')$ by construction. Because $Q, R \notin \mathcal{Q}(a')$, this implies $q_{\bar{z}}(a') > q_{z}(a')$.
We conclude that $\sum_{\bar{a} \in P \cap A} q_{\bar{z}}(\bar{a}) > \sum_{\bar{a} \in P \cap A} q_{z}(\bar{a})$, contradicting our choice of $z$. We have thus shown that $z_P \geq z'_P$, completing the proof of the claim.
\qed
\end{proof}

\subsection{Proof of \cref{lem:computing-mu} (Determining a Basis)}
\label{sec:computing-mu}

\begin{proof}[\cref{lem:computing-mu}]
  Let $\bar{V} \subseteq V$ be the set of nodes such that there is both an $s$-$v$-path and a $v$-$t$-path in $D$, and let $\bar{A} := A[\bar{V}]$ be the set of arcs with both endpoints in $\bar{V}$. Note that the set of $s$-$t$-paths in $D$ is the same as the subgraph $\bar{D} = (\bar{V}, \bar{A})$.
  
  Let $T \subseteq \bar{A}$ be an arborescence rooted at $s$ and spanning $\bar{V}$ (such an arborescence exists by construction of $\bar{V}$). For $v \in \bar{V}$, let $T[v]$ denote the unique $s$-$v$-path in $T$.
  Let $B := \bar{A} \setminus T$. For $b = (v, w) \in B$, let $P_b$ be an $s$-$t$-path resulting from the concatenation of $T[v] \cup {b}$ with an arbitrary $w$-$t$-path in $\bar{D}$ (note that such a concatenation indeed results in a simple path as the graph is acyclic).
  Define $P_0 := T[t]$ and $\bar{\sets} := \{P_0\} \cup \{P_b \st b \in B\}$.
  We claim $\bar{\sets}$ is indeed the desired set of paths.
  
  For any $P \in \sets$, let $\mathbbm{1}_P \in \mathbb{R}^A$ denote the arc-incidence vector of $P$.
  Let $$U := \operatorname{span} \{\mathbbm{1}_P \st P \in \sets\} \text{ and } C := \operatorname{span} \{u \in \mathbb{R}^{A} \st u \text{ is a circulation in } D\}.$$
  Note that $U \subseteq \{\mathbbm{1}_{P_0} + u \st u \in C\}$ because $u_P := \mathbbm{1}_P - \mathbbm{1}_{P_0}$ is a circulation for any $P \in \sets$.
  It is well-known that $\operatorname{dim} C = |A| - (|V| - 1) = |B|$. Hence, $\operatorname{dim} U \leq \operatorname{dim} U + 1 \leq |B| + 1$.
  Moreover, because the graph is acyclic, there is an ordering $b_1, \dots, b_k$ of $B$ such that $P_{b_i} \cap \{b_{1}, \dots, b_{i-1}\} = \emptyset$ for $i \in \{1, \dots, k\}$.
  Because also $P_0 \cap B = \emptyset$, the incidence vectors of the paths in $\bar{\sets}$ are linearly independent. As $|\bar{\sets}| = |B| + 1 = \operatorname{dim} U$, we conclude that $\bar{\sets}$ is indeed a basis of~$U$. 
  Finally, observe that extending the arc-incidence vectors by node incidences does not change the dimension of $U$, as for every path $P$ these entries are fully determined by the arcs $P$. 
  \qed
\end{proof}

\subsection{Computing $\mu$ Using the Algorithm of \citet{cela2023linear}}
\label{sec:quadratic-mu}

    \citeauthor{cela2023linear} give an explicit description~\citep[Eq.~(4)]{cela2023linear} of a linearization~$\mu' \in \mathbb{R}^A$ of an arbitrary function $\pi' \in \mathbb{R}^\sets$ fulfilling condition~\eqref{eq:conservation}, which they describe in terms of two-path systems~\citep[Proposition~1]{cela2023linear}.
    They also describe how to construct a basis $\bar{\sets}$ of the subspace spanned by the $s$-$t$-paths in $D$, similar to the one described in the proof of \cref{lem:computing-mu}~\citep[Definition~3]{cela2023linear}.
    The basis, the corresponding values of $\pi'_P$ for $P \in \bar{\sets}$, and $\mu'$ via the explicit description can be computed in time $\mathcal{O}(|A|^2 + |A|T_{\pi})$.

    Note that the vector $\mu'$ constructed by~\citet{cela2023linear} might contain negative entries, even when $\pi'$ is nonnegative.
    Still we can apply their results to obtain a solution $\mu$ to \eqref{eq:affine-transform} and~\eqref{eq:affine-nonneg} in time $\mathcal{O}(|A|^2 + |A|T_{\pi})$ as follows.
    We first use the algorithm of~\citep{cela2023linear} to obtain $\mu' \in \mathbb{R}^A$ fulfilling
        $\sum_{a \in P \cap A} \mu'_a = -\pi_P$ for all $P \in \sets$.
    We then define $\mu_a := \mu'_a + \phi_v - \phi_w$ for $a = (v, w) \in A$, where $\phi_v$ for $v \in V \setminus \{t\}$ denotes the shortest-path distance from $s$ to $v$ with respect to $\mu'$, and $\phi_t := -1$.
    Indeed, note that $\sum_{a \in P \cap A} \mu_a = -\phi_t + \sum_{a \in P \cap A} \mu'_a = 1 - \pi_P$. Moreover, note that $\mu_a \geq 0$ for $a = (v, w) \in A$ because $\phi_w \leq \mu'_a + \phi_v$, as the $\phi_v$ are shortest-path distances for $v \neq t$ and $\phi_t = -1 \leq \min_{P \in \sets} \mu'_a = - \max_{P \in \sets} \pi_P$ is at most the shortest-path distances from $s$ to $t$.
    
    Note that the running time of $\mathcal{O}(|A|^2 + |A|T_{\pi})$ is best possible when assuming that a query to the value oracle for $\pi$ requires the $s$-$t$-path to be passed explicitly, as in this case each query may require an arc set of size $\theta(|A|)$ as an argument.
    This does, however, not exclude the possibility that in some special cases, a subquadratic running time can be obtained by encoding $\pi$ in data structures supporting more efficient queries in an amortized sense.
    
\subsection{From Poset to DAG (Hasse diagram)}
\label{sec:hasse}

Note that \cref{thm:conservation-affine} requires $\sets$ to be the set of $s$-$t$-paths in a directed acyclic graph (where each path corresponds to the set of its arcs and nodes), which is a special case of maximal chains in a poset (in that each maximal chain in a DAG consists of an alternating sequence of nodes and arcs).
In this section, we discuss how to apply \cref{thm:conservation-affine} in the generic case where $\sets$ is the set of maximal chains of an arbitrary partially ordered set $(E, \sets)$ and $\pi \in [0, 1]^\sets$ fulfils the conservation law \eqref{eq:conservation}. 

We first observe that \cref{thm:conservation-affine} does not apply directly to the case where~$\sets$ is the set of maximal chains in an arbitrary poset. Indeed, consider the poset $E = \{a_0, a_1, z_0, z_1\}$ with $a_i \prec z_j$ for all $i, j \in \{0, 1\}$ and requirements defined by $\pi_{\{a_0, z_0\}} = 0$ and $\pi_{\{a_0,z_1\}} = \pi_{\{a_1,z_0\}} = \pi_{\{a_1,z_1\}} = 1$.
Note that \eqref{eq:conservation} is fulfilled for $\pi$ (in fact, any requirement vector fulfils \eqref{eq:conservation} for this poset).
However, $\pi$ is not of the form $1 - \sum_{e \in P} \mu_e$ for any $\mu \in [0,1]^E$ because $\pi_{\{a_0,z_1\}} = \pi_{\{a_1,z_0\}} = \pi_{\{a_1,z_1\}} = 1$ implies $\mu_{a_0} = \mu_{a_1} = \mu_{z_0} = \mu_{z1} = 0$, contradicting $\pi_{\{a_0, z_0\}} = 0$.

However, we can still make use of \cref{thm:conservation-affine} when given an arbitrary poset, by using its Hasse diagram representation.
The Hasse diagram of a poset $(E, \preceq)$ is the digraph $D = (V, A)$ with node set $V = E$ and arc set $A = \{(e, f) \st e, f \in E, e \prec f \text{ and there is no } f' \in E \text{ with } e \prec f' \prec f\}$.
Note that~$D$ is acyclic. 
Furthermore, if we assume without loss of generality that $\preceq$ has a unique minimal element $s \in E$ and a unique maximal element $t \in E$, then each $s$-$t$-path $P$ in $D$ corresponds to the maximal chain $C = P \cap V$ of $\preceq$ and vice versa.
We can thus interpret $\pi$ also as a requirement function on the paths of $D$ and note that \eqref{eq:conservation} is not affected by this transformation. 

Now we can apply \cref{thm:conservation-affine} to $D$ and $\pi$ obtaining a representation of $\pi$ of the form $\pi_P = 1 - \sum_{e \in P} \mu_e$ by weights $\mu \in [0, 1]^{V \cup A}$.\footnote{Note that these weights $\mu$ are also defined for the arcs of $D$, hence giving $|A|$ additional degrees of freedom compared to the original poset. Our earlier example shows that these degrees of freedom are indeed necessary to obtain an affine representation of $\pi$.}
Given any marginals $\rho \in [0, 1]^E$ in the poset fulfilling \eqref{eq:cond}, we can thus find a feasible decomposition by applying \cref{thm:computation} to the corresponding marginals in $D$, which arise from extending $\rho$ by $\rho_a = 0$ for all $a \in A$. Note that the corresponding decomposition only uses subsets of $V = E$ as $\rho_a = 0$ for $a \in A$, and hence is also a feasible decomposition for $\rho$ in $E$. We summarize these observations in the following corollary.

\begin{corollary}
    There is an algorithm that, given a poset $(E, \preceq)$, a value oracle for a function $\pi \in [0, 1]^\sets$ fulfilling \eqref{eq:conservation}, where $\sets$ is the set of maximal chains in $(E, \preceq)$, and $\rho \in [0, 1]^E$ fulfilling \eqref{eq:cond}, computes a feasible decomposition of $\rho$ with respect to $(E, \sets, \pi)$ in time polynomial in $|E|$ and $T_{\pi}$, where $T_{\pi}$ denotes the time for a call to the value oracle.
\end{corollary}

\section{\texorpdfstring{Feasible Decompositions in the Absence of \eqref{eq:cond}-Sufficiency}{}}
\label{sec:other-systems}

In this final section, we turn to the question whether we can obtain feasible decompositions in systems not exhibiting \eqref{eq:cond}-sufficiency.
We first observe that for any marginals fulfilling~\eqref{eq:cond}, at least an \emph{approximately feasible} decomposition always exist and can be easily constructed.
However, we also show that determining the feasibility of marginals in arbitrary systems is NP-complete in general.

\subsection{Approximately Feasible Decompositions}
We say that a decomposition $x$ of marginals $\rho$ is \emph{$\beta$-approximately feasible}, for a scalar $\beta \in [0, 1]$, if it fulfils \eqref{eq:consistency}, \eqref{eq:simplex}, \eqref{eq:nonneg}, and 
\begin{align*}
    \elsum{S \subseteq E : S \cap P \neq \emptyset} x_S \;\geq\; \beta \cdot \pi_P \quad \forall\, P \in \sets.
\end{align*}
Indeed, for $\beta = 1 - \nicefrac{1}{\text{e}}$ such a decomposition exists for any set of marginals fulfilling \eqref{eq:cond}, without imposing any further conditions on $(E, \sets)$ or $\pi$, as shown by the theorem below.
An interesting question for future research is whether better guarantees can be obtained for particular classes of set systems.

\begin{theorem}
    Let $(E, \sets)$ be a set system, $\pi \in [0,1]^\sets$, and $\rho \in [0, 1]^E$ such that condition \eqref{eq:cond} is fulfilled.
    Let $S_{\rho} \subseteq E$ be the random set that contains each~$e \in E$ independently with probability $\rho_e$ and let $x_S = \prob{S_{\rho} = S}$. Then~$x$ is a $(1 - \nicefrac{1}{e})$-approximately feasible decomposition of $\rho$ for $(E, \sets, \pi)$.
\end{theorem}
\begin{proof}
    By construction, $x$ fulfils \eqref{eq:consistency}, \eqref{eq:simplex}, and \eqref{eq:nonneg}.
    Let $P \in \sets$. Using that $\sum_{e \in P} \rho_e \geq \pi_P$ and $\pi_P \in [0, 1]$, we observe that
    \begin{align*}
        \prob{S_\rho \cap P \neq \emptyset} & = 1 - \prod_{e \in P} \prob{e \notin S_{\rho}} = 1 - \prod_{e \in P} (1 - \rho_e)\\
        & \geq 1 - \left(1 - \tfrac{\sum_{e \in P} \rho_e}{|P|}\right)^{|P|} \geq 1 - \left(1 - \tfrac{\pi_P}{|P|}\right)^{|P|}\\
        & \geq 1 - \text{e}^{-\pi_P} \geq (1 - \text{e}^{-1}) \pi_P,
    \end{align*}
    which proves the lemma. \qed
\end{proof}

\subsection{Hardness of Determining Feasibility of Marginals}
For a given triplet $(E, \sets, \pi)$ and marginals $\rho$, we may also be interested in finding a decomposition that is $\beta$-approximately feasible for the largest possible value of $\beta$. 
Note that solving this problem includes, in particular, deciding whether  $\rho$ has a feasible decomposition (i.e., $\beta = 1$).
Unfortunately, this latter problem is NP-complete, even in quite restricted cases, as evidenced by the theorem below. 

\begin{restatable}{theorem}{restateThmHardness}
    The following decision problem is NP-complete: Given a set system $(E, \sets)$ with $|P| = 3$ for all $P \in \sets$ and marginals $\rho \in [0, 1]^E$, is there a feasible decomposition of $\rho$ for $(E, \sets)$ and requirement vector $\pi \equiv 1$?
 \end{restatable}
    
\begin{proof}
    To see that the problem is indeed in NP, note that a given instance is a yes instance if and only if the corresponding linear system \eqref{eq:consistency} to \eqref{eq:nonneg} has a feasible solution.
    Note that all extreme point solutions of this linear system have at most a polynomial number of non-zeros as the number of constraints \eqref{eq:consistency} to \eqref{eq:simplex} is polynomial in $|E|$ (because $|\sets| \leq |E|^3$). Such an extreme point solution is thus a polynomially sized and verifiable certificate for a yes instance.
    
    We now show completeness via reduction from the NP-complete \emph{Not-all-equal 3-Satisfiability} (NAE3SAT) problem~\citep{schaefer1978complexity}.
    In this problem, we are given a finite set of boolean variables $Y$ and a finite set of clauses $\mathcal{C}$, with each clause consisting of the disjunction of three literals of the variables.
    We let $L := \{y, \neg y \st y \in Y\}$ be the set of literals on $Y$ and identify each clause $C \in \mathcal{C}$ with the set of its literals (e.g., $C = \{y_1, \neg y_3, y_5\}$).
    A \emph{truth assignment} is a subset of $A \subseteq L$ of the literals  containing for each variable $y \in Y$ either $y$ or $\neg y$. 
    The goal of NAE3SAT is to decide whether there is a truth assignment $A$ such that $|A \cap C| \in \{1, 2\}$ for each $C \in \mathcal{C}$, i.e., the truth assignment satisfies at least one, but not all literals of each clause.

    For a truth assignment $A$, let $\bar{A} := Y \setminus A$ denote the corresponding complementary assignment.
    For a clause $C \in \mathcal{C}$, let $$\bar{C} := \{\neg y \st y \in C \cap Y\} \cup \{y \st \neg y \in C\}$$ denote the corresponding complementary clause consisting of the negation of its literals.
    We will make use of the following observation.
    \begin{observation}\label[observation]{obs:truth-assignment}
        For any truth assignment $A$, the following statements are equivalent:
        \begin{enumerate}
            \item $|A \cap C| \in \{1, 2\}$ for every $C \in \mathcal{C}$.
            \item $A \cap C \neq \emptyset$ and $\bar{A} \cap C \neq \emptyset$ for all $C \in \mathcal{C}$.
            \item $A \cap C \neq \emptyset$ and $A \cap \bar{C} \neq \emptyset$ for all $C \in \mathcal{C}$.
            \item $\bar{A} \cap C \neq \emptyset$ and $\bar{A} \cap \bar{C} \neq \emptyset$ for all $C \in \mathcal{C}$.
        \end{enumerate}
    \end{observation}
    
    Given an instance of NAE3SAT, define the set system $(E, \sets)$ with $E := L$ and
    $$\sets := \{C, \bar{C} \st C \in \mathcal{C}\} \cup \{P_y : y \in Y\},$$
    where $P_y := \{y, \neg y\}$.
    Let $\pi \equiv 1$ and $\rho_e = 0.5$ for all $e \in E$.
    We show that there is a feasible decomposition of $\rho$ for $(E, \sets)$ and $\pi$ if and only if there is a truth assignment $A$ of the variables in $Y$ with $|A \cap C| \in \{1, 2\}$ for all $C \in \mathcal{C}$.
    
    First assume there is a truth assignment $A$ of the variables in $Y$ with $|A \cap C| \in \{1, 2\}$ for all $C \in \mathcal{C}$.
    Then define $x$ by setting $x_A = 0.5$, $x_{\bar{A}} = 0.5$, and $x_S = 0$ for all $S \in 2^E \setminus \{A, \bar{A}\}$.
    Note that $x$ fulfils \eqref{eq:simplex} and \eqref{eq:nonneg} by construction.
    Moreover, because $A, \bar{A}$ is a partition of $E = L$, note that $x$ fulfils \eqref{eq:consistency} and $\sum_{S : S \cap \{y, \neg y\} \neq \emptyset} x_S = x_A + x_{\bar{A}} = 1$ for all $y \in Y$. 
    Finally, because $|A \cap C| \in \{1, 2\}$, note that $A \cap C, A \cap \bar{C}, \bar{A} \cap C, \bar{A} \cap \bar{C}$ are all nonempty for every $C \in \mathcal{C}$ by \cref{obs:truth-assignment}, and hence $\sum_{S : S \cap C} x_S = x_A + x_{\bar{A}} = 1$ and $\sum_{S : S \cap \bar{C}} x_S = x_A + x_{\bar{A}} = 1$ for all $C \in \mathcal{C}$.
    Thus $x$ also fulfils \eqref{eq:covering} and we can conclude that $x$ is a feasible decomposition of $\rho$.
    
    Now assume $\rho$ has a feasible decomposition $x$. Let $S \subseteq E$ with $x_S > 0$.  Note that 
    $$1 \;=\; \pi_{P_y} \;\leq\; \elsum{S : S \cap P_y \neq \emptyset} x_S \;\leq\; \elsum{S : S \cap P_y \neq \emptyset} |S \cap P_y| x_S \;\leq\; \elsum{S \subseteq E} x_S \;\leq\; 1,$$ 
    for all $y \in Y$, and hence 
    $|S \cap P_y| = 1$ for all $S$ with $x_S > 0$.
    Thus also $|S \cap P_y| = 1$ for all $y \in Y$, i.e., $S$ contains exactly one of the literals $y, \neg y$ for every $y \in Y$. 
    Hence $S$ is a truth assignment.
    Similarly, note that 
    \begin{align*}
        1 \;\leq\; \elsum{S : S \cap C \neq \emptyset} x_S \;\leq\; \sum_{S \subseteq E} x_S \;\leq\; 1 \quad \text{and} \quad 1 \;\leq\; \elsum{S : S \cap \bar{C} \neq \emptyset} x_S \;\leq\; \sum_{S \subseteq E} x_S \;\leq\; 1
    \end{align*} 
    for all $C \in \mathcal{C}$, and hence $x_S > 0$ implies $S \cap C \neq \emptyset \neq S \cap \bar{C}$. Hence, by \cref{obs:truth-assignment}, $S$ is a truth assignment fulfilling $|S \cap C| \in \{1, 2\}$ for all~\mbox{$C \in \mathcal{C}$}.~\qed
\end{proof}

 \section*{Statements and Declarations}
 
 \paragraph{Acknowledgements.} 
 This work was supported by the special research fund of KU~Leuven (project C14/22/026) and by the Center for Mathematical Modeling at Universidad de Chile (Grant ANID FB210005).
 The author thanks three anonymous reviewers of the conference version for numerous helpful suggestions that improved the manuscript. 

\paragraph{Competing Interests.}
The author confirms to have no conflicts of interest associated with this publication and to have not received any financial support that could have influenced the outcome of this work. 

\bibliographystyle{plainnat}
\bibliography{decomposition}
    
\end{document}